\documentclass[lettersize,journal]{IEEEtran}

\usepackage[T1]{fontenc}

\usepackage{bm}
\usepackage{amsthm}
\usepackage{amsmath}
\newtheorem{lemma}{Lemma}
\newtheorem{definition}{Definition}
\newtheorem{theorem}{Theorem}

\usepackage{makecell}

\usepackage{cite}

\usepackage[pdftex]{graphicx}

\ifCLASSINFOpdf

\else

\fi

\usepackage{diagbox}
\usepackage{xcolor}

\usepackage{amsmath}

\usepackage[cmintegrals]{newtxmath}

\usepackage{algorithm}

\usepackage{algcompatible}

\usepackage{hyperref}
\usepackage{subfigure}
\usepackage{makecell}

\begin{document}

\title{Joint Activity-Delay Detection and Channel Estimation for Asynchronous Massive Random Access: A Free Probability Theory Approach}

\author{Xinyu~Bian,~\IEEEmembership{Graduate Student Member,~IEEE,}
Yuyi~Mao,~\IEEEmembership{Member,~IEEE,}
and~Jun~Zhang,~\IEEEmembership{Fellow,~IEEE}
\thanks{X. Bian and J. Zhang are with the Department of Electronic and Computer Engineering, the Hong Kong University of Science and Technology, Hong Kong (E-mail: xinyu.bian@connect.ust.hk, eejzhang@ust.hk). Y. Mao is with the Department of Electrical and Electronic Engineering, the Hong Kong Polytechnic University, Hong Kong (E-mail: yuyi-eie.mao@polyu.edu.hk). This work was supported by the General Research Fund (Project No. 15207220) from the Hong Kong Research Grants Council. \emph{(Corresponding author: Yuyi Mao)}

Part of this work was presented at the 2023 IEEE Global Communications Conference (GLOBECOM) \cite{xbian2023conference}.
}
\thanks{}
\thanks{}}

\markboth{}%
{Shell \MakeLowercase{\textit{et al.}}: Bare Demo of IEEEtran.cls for IEEE Communications Society Journals}

\maketitle

\begin{abstract}
Grant-free random access (RA) has been recognized as a promising solution to support massive connectivity due to the removal of the uplink grant request procedures. While most endeavours assume perfect synchronization among users and the base station, this paper investigates asynchronous grant-free massive RA, and develop efficient algorithms for joint user activity detection, synchronization delay detection, and channel estimation. Considering the sparsity on user activity, we formulate a sparse signal recovery problem and propose to utilize the framework of orthogonal approximate message passing (OAMP) to deal with the non-independent and identically distributed (i.i.d.) Gaussian pilot matrices caused by the synchronization delays. In particular, an OAMP-based algorithm is developed to fully harness the common sparsity among received pilot signals from multiple base station antennas. To reduce the computational complexity, we further propose a free probability AMP (FPAMP)-based algorithm, which exploits the rectangular free cumulants to make the cost-effective AMP framework compatible to general pilot matrices. Simulation results demonstrate that the two proposed algorithms outperform various baselines, and the FPAMP-based algorithm reduces $40\%$ of the computations while maintaining comparable detection/estimation accuracy with the OAMP-based algorithm.
\end{abstract}
\begin{IEEEkeywords}
Grant-free massive random access, activity detection, delay detection, channel estimation, asynchronous connectivity, free probability theory, approximate message passing (AMP).
\end{IEEEkeywords}

\IEEEpeerreviewmaketitle

\section{Introduction \label{sectioni}}
As applications supported by massive machine-type communications (mMTC) technologies rapidly evolve \cite{cbo2016}, an enormous challenge of supporting reliable communication for massive devices with limited bandwidth resources arises \cite{yshi2020}. This challenge should be tackled together with other unique features of mMTC, such as the sporadic data traffic pattern, via novel random access (RA) schemes \cite{mt2014}. Conventionally, each user (e.g. a Internet-of-Things device) requests a grant from the base station (BS) for its transmission by initiating a complex handshaking process. Such grant-based RA schemes shall result in significant signaling overhead and access latency \cite{psch2017,xchen2021}, especially in the target scenarios with massive potential users. In addition, due to the limited preamble sequences for uplink grant request, grant-based RA suffers from potential collisions when two or more users pick the same sequence \cite{ebj2017}. As a result, grant-free RA, where each active user directly transmits data without waiting for access permissions from the BS, is more appealing for mMTC \cite{jchoi2022}.

Since the user activity is unknown beforehand at the BS with grant-free RA, each user is assigned with a fixed and unique pilot sequence to enable user activity detection and channel estimation, which is critical to the downstream data detection and decoding \cite{ywu2020}. However, attributed to the large number of users and the limited resources for pilot transmission, only non-orthogonal pilot sequences can be adopted, which unavoidably compromises the accuracy of user activity detection and channel estimation. Therefore, many works have been carried out to develop advanced user activity detection and channel estimation algorithms for grant-free RA systems.

\subsection{Related Works and Motivations}
User activity detection and channel estimation have been the hot spots in the research area of grant-free RA. In \cite{zchen2019} and \cite{sha2018}, the set of active users are determined by solving a maximum likelihood (ML) estimation problem via coordinate-wise descent according to the sample covariance-matrix of received pilot signal. Joint activity detection and channel estimation (JADCE) was investigated in \cite{zchen2018,mke2020,xbian2023, xbian2023tcom, ycui2021} by adopting the family of approximate message passing (AMP) algorithms. In particular, by formulating JADCE as a single measurement vector (SMV) and multiple measurement vector (MMV) compressive sensing problem with single- and multi-antenna BS, respectively, a variety of AMP algorithms with a minimum mean square error (MMSE) denoiser were developed in \cite{zchen2018}. To improve the user activity detection and channel estimation accuracy, wireless channel sparsity from the spatial and angular domain were further utilized by a general MMV-AMP algorithm in \cite{mke2020}. Besides, the common sparsity pattern in the received pilot and data signal was exploited via the bilinear generalized AMP (BiG-AMP) algorithm along with soft data decoding information in \cite{xbian2023}, while a correlated AMP algorithm were developed to take advantages of user activity correlation among transmission blocks in \cite{xbian2023tcom}. Moreover, low-complexity JADCE methods were proposed based on deep learning in \cite{yqiang2020, ycui2021}.

Nevertheless, the above studies assume that users and the BS are perfectly synchronized, which is impractical in grant-free RA systems. This is because without coordination of the BS, different users may send pilot sequences on-demand, which can be randomly delayed by some unknown symbol periods \cite{hongyuan2008}. Therefore, asynchronous grant-free massive RA has attracted growing interest. Specifically, the joint activity and delay detection problem in asynchronous grant-free RA was formulated as a group least absolute shrinkage and selection operator (LASSO) problem in \cite{lliu2021}, which was solved by a block coordinate descent algorithm. However, this algorithm does not leverage the common sparsity among the received pilot signals of multiple BS antennas. Accordingly, a sample covariance matrix-based approach was proposed in \cite{zwang2022} to achieve higher accuracy of activity and delay detection at increased computational complexity. Note that the methods in \cite{lliu2021} and \cite{zwang2022} require separate channel estimation procedures following user activity and delay detection. To overcome this limitation, joint activity-delay detection and channel estimation was investigated in \cite{yguo2023} and \cite{mqiu2022}, where a low-complexity orthogonal matching pursuit (OMP)-based and alternating direction method of multipliers (ADMM)-based methods were developed, respectively. However, it can be anticipated that both methods are far from optimal. On one hand, the OMP-based algorithm is not able to utilize prior knowledge of the channel coefficients. On the other hand, the ADMM-based approach does not explicitly incorporate the synchronization delays in optimizations. Therefore, \cite{wzhu2021} proposed a model-driven learned AMP (LAMP) network for asynchronous grant-free massive RA, which unrolled the AMP iterations as neural network layers. However, since entries of the effective pilot matrices in asynchronous grant-free massive RA systems are not independent and identically distributed (i.i.d.) Gaussian due to the delayed pilot symbols, such a basic AMP-based algorithm cannot achieve the best detection and estimation accuracy. It is thus of pressing needs to develop advanced receivers for asynchronous grant-free massive RA to enhance the performance of activity-delay detection and channel estimation.

\subsection{Contributions}
In this paper, we develop novel joint activity-delay detection and channel estimation algorithms for asynchronous grant-free massive RA systems, which enjoy premium performance and low computational complexity. Our main contributions are summarized below.

\begin{itemize}
    \item Given that the conventional AMP-based joint activity detection and channel estimation algorithm fails to handle pilot matrices with non-i.i.d. Gaussian entries caused by the synchronization delay, the framework of orthogonal AMP (OAMP) \cite{jma2017} is adopted to jointly estimate the user activity, synchronization delay, and channel coefficients for asynchronous grant-free massive RA systems. To boost the detection/estimation performance, the common sparsity among multiple BS antennas is also utilized to update the prior information in each iteration of the OAMP-based algorithm.
\end{itemize}

\begin{itemize}
    \item One drawback of the OAMP-based algorithm is the high computational complexity due to the use of a linear MMSE (LMMSE) estimator. In order to reduce the complexity, we resort to the free probability theory and develop a free probability AMP (FPAMP)-based algorithm for joint activity-delay detection and channel estimation. Unlike the OAMP-based algorithm, FPAMP estimates the moments of the eigenvalue distributions of the pilot matrix, and exploits its rectangular free cumulants to extend the low-complexity AMP framework for general pilot matrices, which avoids using the complex LMMSE estimator.
\end{itemize}

\begin{itemize}
    \item Simulation results show that the proposed OAMP-based and FPAMP-based algorithms significantly reduce the activity-delay detection and channel estimation errors compared with the baseline schemes. In particular, if the probability of false alarm and missed detection requirement are respectively $10^{-1}$ and $2\times 10^{-3}$, the two proposed algorithms are able to support $76\%$ more active users in comparsion with the conventional AMP-based algorithm. In addition, the FPAMP-based algorithm enjoys a similar complexity as the AMP-based algorithm, while maintaining almost the same detection/estimation performance as the OAMP-based algorithm.
\end{itemize}

\subsection{Organization}
The rest of this paper is organized as follows. We introduce the system model in Section \ref{sectionii}. In Section \ref{sectioniii}, we propose an OAMP-based algorithm for joint activity-delay detection and channel estimation. A low-complexity FPAMP-based algorithm is developed in Section \ref{sectioniv}. Simulation results are presented in Section \ref{sectionv} and conclusions are drawn in Section \ref{sectionvi}.

\subsection{Notations} 
We use lower-case letters, bold-face lower-case letters, bold-face upper-case letters, and math calligraphy letters to denote scalars, vectors, matrices, and sets, respectively. The transpose and conjugate transpose of matrix $\mathbf{M}$ are denoted as $\mathbf{M}^{\mathrm{T}}$ and $\mathbf{M}^{\mathrm{H}}$, respectively. Besides, we denote the complex Gaussian distribution with mean $\bm{\mu}$ and covariance matrix $\bm{\Sigma}$ as $\mathcal{C} \mathcal{N}(\bm{\mu}, \bm{\Sigma})$, and the probability density function (PDF) of a complex Gaussian variable $\bm{x}$ as $\mathcal{C} \mathcal{N}(\bm{x};\bm{\mu}, \bm{\Sigma})$. In addition, $\mathbf{I}_{L}$ denotes the $L$-by-$L$ identify matrix, $\delta_{0}$ denotes the Dirac delta function, “$\otimes$” stands for the Kronecker product, $\operatorname{tr}(\cdot)$ returns the trace of a matrix, and $\mathbb{E}[\cdot]$ and $\operatorname{Var}[\cdot]$ denote the statistical expectation and variance, respectively. Given vector $\mathbf{x}=(x_1,\cdots,x_L)^{\mathrm{T}}$, we use $\left\langle\mathbf{x}\right\rangle$ to denote its empirical average, i.e., $\frac{1}{L}\sum_{i=1}^{L}x_{i}$. We further use $[\mathbf{M}]_{i_1:i_2,j_1:j_2}$ to denote the submatrix obtained by extracting elements of matrix $\mathbf{M}$ from row $i_1$ to $i_2$ and column $j_1$ to $j_2$. If $i_1=i_2$ or $j_1=j_2$, the submatrix reduces to a vector and the second index is omitted. Moveover, $[f(\cdot)]_{t}$ denotes the $t$-th entry of the output of vector-valued function $f(\cdot)$.

\section{System Model \label{sectionii}}
We consider an uplink asynchronous grant-free massive RA system with $N$ single-antenna users and an $M$-antenna BS, where the set of users and BS antennas are denoted as $\mathcal{N}\triangleq \{1,\cdots,N\}$ and $\mathcal{M}\triangleq\{1,\cdots,M\}$, respectively. Users are assumed to have sporadic data traffic. Specifically, at each channel block, $K$ ($K\leq N$) of the $N$ users independently become active for uplink transmission with equal probability \cite{xbian2023}. Let $u_{n} \in \{0,1\}$ be the activity indicator of user $n$, where $u_{n}=1$ indicates that user $n$ is active and vice versa. Denote the set of active users as $\mathcal{K} \triangleq \left\{n \in \mathcal{N} | u_{n}=1 \right\}$. We assume quasi-static Rayleigh fading channels for simplicity, and denote the uplink channel vector from user $n$ to the BS as $\mathbf{f}_{n}\sim \mathcal{CN}(\mathbf{0},\beta_{n}\mathbf{I}_{M})$, where $\beta_{n}$ is the path loss of user $n$ known at the BS.

A grant-free RA scheme is adopted for uplink transmission, which divides each transmission block into two phases for transmitting pilot and data signals, respectively. In the first phase, $\bar{L}$ pilot symbols are transmitted for user activity detection and channel estimation at the BS. Since orthogonal pilot sequences are insufficient for massive potential users, a unique and non-orthogonal pilot sequence, denoted as $\sqrt{\bar{L}}\bar{\mathbf{p}}_{n}$, where $\bar{\mathbf{p}}_{n}\triangleq \left[\bar{p}_{n, 1}, \cdots, \bar{p}_{n, \bar{L}}\right]^{\mathrm{T}}$ and $\bar{p}_{n, l} \sim \mathcal{C} \mathcal{N}\left(0, 1/\bar{L}\right)$, is assigned to each user \cite{zchen2018}. However, attributed to the synchronization delay, signal transmissions of different active users may start at different time instant. Without loss of generality, we assume each user transmits asynchronously in the frame level, but synchronously in the symbol level, i.e., the pilot sequence is transmitted with a delay of some unknown symbol periods. We use $t_{n}$ to denote the unknown delay for user $n$, which is assumed to be an integer uniformly distributed in set $\{0,\cdots,T\}$ with $T$ denoting the maximum synchronization delay. To avoid interference between pilot and data signals, we insert a guard interval spanning $T$ symbol periods between them following \cite{wzhu2021}, as shown in Fig. \ref{system}. Therefore, the expanded pilot sequence of user $n$ with delay $t_{n}$, denoted as $\tilde{\mathbf{p}}_{n,t_{n}}$, can be expressed as $\tilde{\mathbf{p}}_{n,t_{n}}\triangleq\left[\mathbf{0}_{t_{n}}^{\mathrm{T}},\bar{\mathbf{p}}_{n}^{\mathrm{T}},\mathbf{0}_{T-t_{n}}^{\mathrm{T}}\right]^{\mathrm{T}}$, which is a sequence with length $L=\bar{L}+T$ obtained by padding $t_{n}$ and $T-t_{n}$ zeros before and after $\bar{\mathbf{p}}_{n}$, respectively. Correspondingly, by considering all possible synchronization delays, we define $\mathbf{P}_{n}\triangleq \left[\tilde{\mathbf{p}}_{n,0},\cdots,\tilde{\mathbf{p}}_{n,T}\right] \in \mathbb{C}^{L\times (T+1)}$ as the expanded pilot matrix of user $n$. 
\begin{figure}[t]
\centering
\includegraphics[width=3.4in]{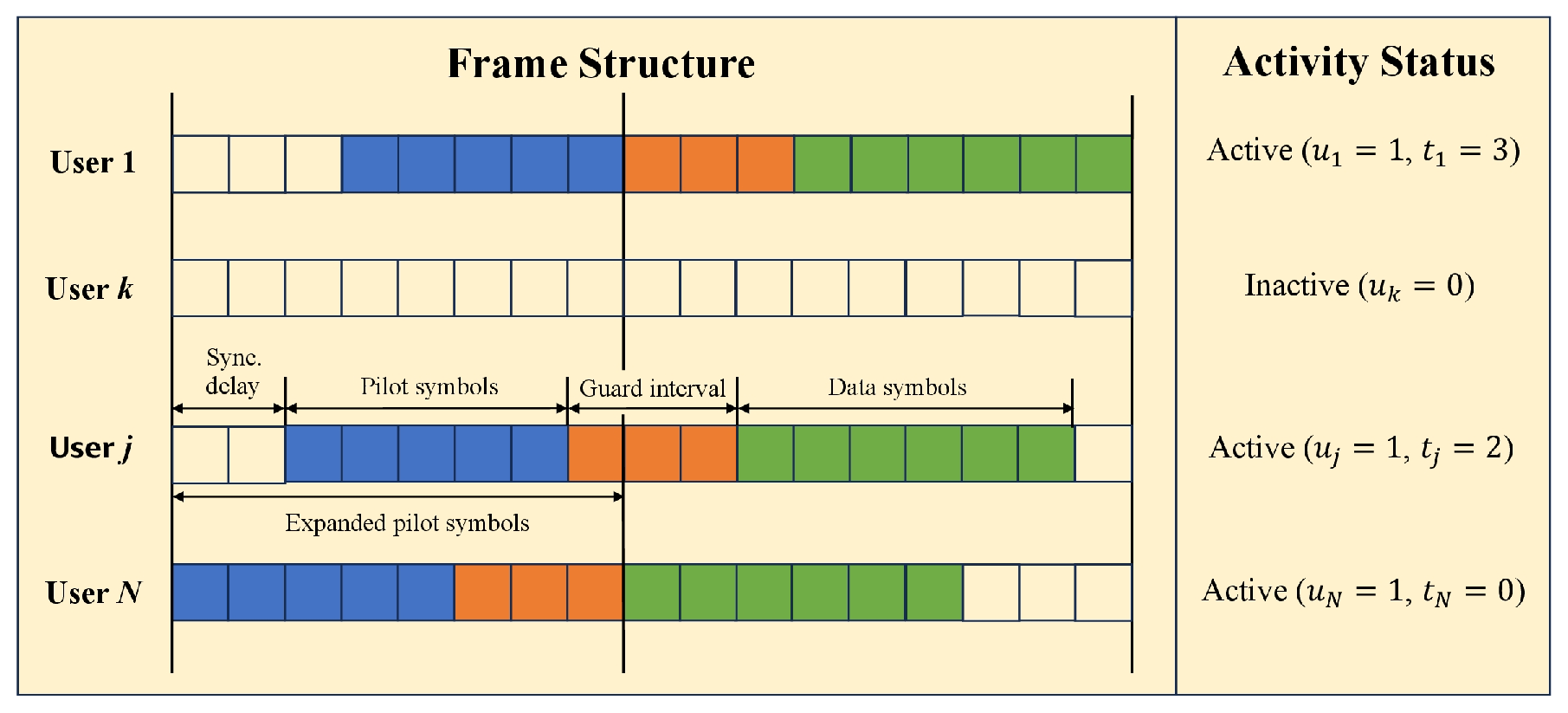}
\caption{Frame structure of asynchronous grant-free massive RA, where a guard interval spanning $T=3$ symbol periods are inserted between pilot and data symbols.}
\label{system}
\end{figure}

Since both the user activity and synchronization delay are unknown at the BS, we further define the expanded indicator $\boldsymbol{\phi}_{n}\triangleq [\phi_{n,0},\cdots,\phi_{n,T}]^{\mathrm{T}}$ for user $n$, where $\phi_{n,t}=1$ only when $u_{n}=1$ and $t_{n}=t$. The received pilot signal $\tilde{\mathbf{Y}} \in \mathbb{C}^{L\times M}$ at the BS is expressed as follows: 
\begin{align}
\tilde{\mathbf{Y}}&=\sqrt{\rho L}\mathbf{P}\mathbf{H}+\tilde{\mathbf{N}}, \label{eqsystemmodel}
\end{align}
\noindent where $\mathbf{P}\triangleq [\mathbf{P}_{1},\cdots,\mathbf{P}_{N}] \in \mathbb{C}^{L\times(T+1)N}$ concatenates the expanded pilot matrices of all users, and $\mathbf{H} \triangleq \left[\mathbf{H}_{1}^{\mathrm{T}},...,\mathbf{H}_{N}^{\mathrm{T}}\right]^{\mathrm{T}} \in \mathbb{C}^{(T+1)N\times M}$ stands for the expanded effective channel matrix with $\mathbf{H}_{n}\triangleq \boldsymbol{\phi}_{n}\otimes \mathbf{f}_{n} \in \mathbb{C}^{(T+1)\times M}$. Besides, $\rho$ is the transmit power, and $\tilde{\mathbf{N}}=\left[\tilde{\mathbf{n}}_{1},...,\tilde{\mathbf{n}}_{L}\right]^{\mathrm{T}}$ denotes the Gaussian noise with zero mean and variance $\sigma^{2}$ for each element. We normalize the received signal and noise as $\mathbf{Y}\triangleq \tilde{\mathbf{Y}} \slash \sqrt{\rho L}$ and $\mathbf{N}\triangleq \tilde{\mathbf{N}} \slash \sqrt{\rho L}$ respectively for notational convenience.

In contrast to synchronous grant-free massive RA systems where only the user activity and channel coefficients need to be estimated, the synchronization delays of active users warrant additional attention in asynchronous grant-free massive RA systems. However, the heterogeneous synchronization delays among active users breaks the standard assumption that entries of the expanded pilot matrix are independent and follow an identical Gaussian distribution, which makes existing AMP algorithms incompetent. In the following sections, we will develop efficient algorithms to detect the set of active users, and estimate their synchronization delays as well as channel coefficients.

\section{OAMP-based Algorithm for Asynchronous Massive RA \label{sectioniii}}
To tackle the limitation of AMP in handling the non-i.i.d. entries in the expanded pilot matrix $\mathbf{P}$, the framework of OAMP \cite{jma2017} emerges as an ideal candidate. However, the original OAMP framework was proposed for SMV problems, which cannot be directly applied to asynchronous grant-free RA systems with multi-antenna BSs. Therefore, we first derive the operations in an OAMP iteration based on the received signal of individual BS antennas in Section \ref{SMV}, where a denoiser is carefully designed according to prior information of the expanded effective channel matrix. Then, we further develop the OAMP iterations by utilizing the common sparsity among received signals of multiple BS antennas \cite{ymei2022}, i.e., enforcing the same prior sparsity ratio for all antennas, in Section \ref{extendtoMMV}, leading to the OAMP-based algorithm for asynchronous grant-free massive RA.

\subsection{OAMP Iteration for Each BS Antenna \label{SMV}}
Since both the user activity and synchronization delay are embedded
in $\mathbf{H}$ and can be determined accordingly, we first derive the operations in an OAMP iteration for each BS antenna based on its normalized received pilot signal given as follows:
\begin{align}
\mathbf{y}_{m}=\mathbf{P}\mathbf{h}_{m}+\mathbf{n}_{m}, \forall m \in \mathcal{M}, 
\label{SMVmodel}
\end{align}
\noindent where $\mathbf{y}_{m}$, $\mathbf{h}_{m}$, and $\mathbf{n}_{m}$ denote the $m$-th column of $\mathbf{Y}$, $\mathbf{H}$, and $\mathbf{N}$, respectively. The conventional OAMP algorithm iterates between a linear estimator (LE) and a non-linear estimator (NLE). In principle, by restricting a de-correlated estimator in LE and a divergence-free estimator in NLE, the input and output errors for both LE and NLE are orthogonal so that the OAMP algorithm achieves Bayes-optimal performance. Starting with $\mathbf{s}_{m}^{(1)}=\mathbf{0}$, the operations in the $i$-th OAMP iteration is expressed as follows:
\begin{align}
\text{LE:} \ \ \ \ \ \ \mathbf{r}_{m}^{(i)}=\mathbf{s}_{m}^{(i)}+\mathbf{W}_{m}^{(i)}(\mathbf{y}_{m}-\mathbf{P}\mathbf{s}_{m}^{(i)}),\label{LE}
\end{align}
\begin{align}
\begin{aligned}
\text{NLE:} \ \ \ \mathbf{s}_{m}^{(i+1)}&=\!\eta^{(i)}(\mathbf{r}_{m}^{(i)})\\
&=\!C_{m}^{(i)}\Big(\hat{\eta}^{(i)}(\mathbf{r}_{m}^{(i)})\!-\!\frac{\sum_{n=1}^{N}\sum_{t=0}^{T}[\hat{\eta}^{\prime(i)}(\mathbf{r}_{n,m}^{(i)})]_{t}}{(T+1)N}\mathbf{r}_{m}^{(i)}\Big),\label{NLE}
\end{aligned}
\end{align}
\noindent where $\mathbf{r}^{\left(i\right)}_{m}$ and $\mathbf{s}_{m}^{\left(i+1\right)}$ are respectively the output of the LE and NLE, and other notations will be introduced in the sequel.

\subsubsection{LE}
The LE aims at decorrelating the vector estimation problem in (\ref{SMVmodel}) to $N$ scalar estimation problems for each user, which is achieved by restricting the de-correlated estimator $\mathbf{W}_{m}^{(i)}$ as the following the optimal structure:
\begin{align}
\mathbf{W}_{m}^{(i)}=\frac{(T+1) N}{\operatorname{tr}\left(\hat{\mathbf{W}}_{m}^{(i)} \mathbf{P}\right)} \hat{\mathbf{W}}_{m}^{(i)},
\label{decor}
\end{align}
\noindent where $\hat{\mathbf{W}}_{m}^{(i)}\!=\!\mathbf{P}^{\mathrm{H}}\left(\mathbf{P} \mathbf{P}^{\mathrm{H}}+(\sigma^2/\rho L (v_{m}^{(i)})^{2}) \mathbf{I}_{L}\right)^{-1}$ is the LMMSE estimator with $(v_{m}^{(i)})^{2}$ representing the mean square error between the output of the NLE and the actual expanded channel vector $\mathbf{h}_{m}$ defined as
\begin{align}
(v_{m}^{(i)})^{2}\triangleq\frac{\mathbb{E}[||\mathbf{s}_{m}^{(i)}-\mathbf{h}_{m}||_{2}^{2}]}{(T+1)N}.
\label{errorv}
\end{align}
\noindent The empirical evaluation of $(v_{m}^{(i)})^{2}$ will be presented in (\ref{v_calculate}) when introducing the NLE. Similarly, the mean square error between the output of the LE and the actual expanded channel vector $\mathbf{h}_{m}$, i.e., $(\tau_{m}^{(i)})^{2}\triangleq\frac{\mathbb{E}[||\mathbf{r}_{m}^{(i)}-\mathbf{h}_{m}||_{2}^{2}]}{(T+1)N}$, can be empirically evaluated as follows:
\begin{align}
(\tau_{m}^{(i)})^{2} \approx \frac{\operatorname{tr}\left(\mathbf{B}_{m}^{(i)} (\mathbf{B}_{m}^{(i)})^{\mathrm{H}}\right) (v_{m}^{(i)})^2+\operatorname{tr}\left(\mathbf{W}_{m}^{(i)} (\mathbf{W}_{m}^{(i)})^{\mathrm{H}}\right) \cdot \frac{\sigma^2}{\rho L}}{(T+1) N} ,
\label{errortau}
\end{align}
\noindent where $\mathbf{B}_{m}^{(i)}\triangleq \mathbf{I}_{(T+1)N}-\mathbf{W}_{m}^{(i)}\mathbf{P}$.

\subsubsection{NLE} The NLE applies a denoiser $\eta^{(i)}(\cdot)$ to the output of the LE, which is restricted to be divergence-free, i.e., $\mathbb{E}\left[\eta^{\prime(i)}(\cdot)\right] = 0$ with $\eta^{\prime(i)}(\cdot)$ denoting the ﬁrst-order derivative of $\eta^{(i)}(\cdot)$. Therefore, $\hat{\eta}^{(i)}(\cdot)$ in the optimal NLE is the section-wise MMSE denoiser given as follows:
\begin{align}
\hat{\eta}^{(i)}(\mathbf{r}_{n,m}^{(i)})&=\mathbb{E}[\mathbf{h}_{n,m} \mid \mathbf{r}_{n,m}^{(i)}],
\label{MMSEmean}
\end{align}
\noindent where $\mathbf{h}_{m}$ and $\mathbf{r}_{m}$ contain $N$ sections, i.e., $\mathbf{h}_{m}=[\mathbf{h}_{1,m}^{\mathrm{T}},\cdots,\mathbf{h}_{N,m}^{\mathrm{T}}]^{\mathrm{T}}$ and $\mathbf{r}_{m}^{(i)}=[(\mathbf{r}_{1,m}^{(i)})^{\mathrm{T}},\cdots,$ $(\mathbf{r}_{N,m}^{(i)})^{\mathrm{T}}]^{\mathrm{T}}$. This means that the estimation result of the denoiser in the $i$-th OAMP iteration is the posterior mean of $\mathbf{h}_{n,m}$ given $\mathbf{r}_{n,m}^{(i)}$. Thus, the posterior variance of each entry in $\mathbf{h}_{n,m}$, i.e., $h_{n,t,m}$, is given as follows:
\begin{align}
\psi_{n,t,m}^{(i)}=\mathbb{E}[||[\hat{\eta}^{(i)}(\mathbf{r}_{n,m}^{(i)})]_{t}-h_{n,t,m}||^{2}].
\label{MMSEvar}
\end{align}

\noindent Besides, $C_{m}^{(i)}$ in the optimal NLE, i.e., (\ref{NLE}), is given as follows:
\begin{align}
C_{m}^{(i)}=\frac{(\tau_{m}^{(i)})^{2}}{(\tau_{m}^{(i)})^{2}-\bar{\psi}_{m}^{(i)}},
\label{Cm}
\end{align}
where $\bar{\psi}_{m}^{(i)}\triangleq \frac{\sum_{n=1}^{N}\sum_{t=0}^{T}\psi_{n,t,m}^{(i)}}{(T+1)N}$. On the other hand, the term $\frac{\sum_{n=1}^{N}\sum_{t=0}^{T}[\hat{\eta}^{\prime(i)}(\mathbf{r}_{n,m}^{(i)})]_{t}}{(T+1)N}$ is derived as
\begin{align}
\frac{\sum_{n=1}^{N}\sum_{t=0}^{T}[\hat{\eta}^{\prime(i)}(\mathbf{r}_{n,m}^{(i)})]_{t}}{(T+1)N}=\frac{\bar{\psi}_{m}^{(i)}}{(\tau_{m}^{(i)})^{2}},
\label{bigterm}
\end{align}

\noindent and $(v_{m}^{(i+1)})^{2}$ can be calculated as
\begin{align}
(v_{m}^{(i+1)})^{2} \approx \left(\frac{1}{\bar{\psi}_{m}^{(i)}}-\frac{1}{(\tau_{m}^{(i)})^{2}}\right)^{-1}.
\label{v_calculate}
\end{align}

It remains to obtain the posterior mean and variance in (\ref{MMSEmean}) and (\ref{MMSEvar}), for which, the posterior distribution of $\mathbf{h}_{m}$ is needed. For this purpose, the prior information and the likelihood function of $\mathbf{h}_{m}$ need first to be obtained. Since all users become active with equal probability and the synchronization delay of an active user is uniformly distributed, for user $n$, there is at most one non-zero element in its expanded channel vector $\mathbf{h}_{n,m}$, i.e., it is either inactive or active and has a particular synchronization delay. As a result, we model the prior information of $\mathbf{h}_{n,m}$ for user $n$ as follows:
\begin{align}
\begin{aligned}
p\left(\mathbf{h}_{n,m}\right)&=\left(1-\sum_{t=0}^{T}\lambda_{n,t,m}\right)\prod_{t=0}^{T}\delta_{0}\left(h_{n, t,m}\right)\\
&+\sum_{t=0}^{T}\left(\lambda_{n,t,m}\mathcal{C N}\left(h_{n,t,m} ; 0, \beta_{n}\right)\prod_{t^{\prime}\neq t}\delta_0(h_{n,t^{\prime},m})\right),
\label{prior}
\end{aligned}
\end{align}

\noindent where $\lambda_{n,t,m}=\frac{K}{(T+1)N}$, $\forall t \in \{0,\cdots,T\}$. Since $\mathbf{r}_{n,m}^{(i)}$ is modeled as an observation of $\mathbf{h}_{n,m}$ with additive white Gaussian noise (AWGN) in the OAMP framework, i.e., $\mathbf{r}_{n,m}^{(i)}=\mathbf{h}_{n,m}+\tau_{m}^{(i)}\mathbf{z}_{n,m}$ with $\mathbf{z}_{n,m}\sim \mathcal{CN}(\mathbf{0},\mathbf{I}_{T+1})$ independent of $\mathbf{h}_{n,m}$, according to the Bayes' theorem, the posterior distribution of $\mathbf{h}_{n,m}$ can be obtained as follows:
\begin{align}
\begin{aligned}
p\left(\mathbf{h}_{n,m}|\mathbf{r}_{n,m}^{(i)}\right)&=\left(1-\sum_{t=0}^{T}\pi_{n,t,m}^{(i)}\right)\prod_{t=0}^{T}\delta_{0}\left(h_{n, t,m}\right)\\
&+\sum_{t=0}^{T}\Bigg(\pi_{n,t,m}^{(i)}\mathcal{C N}\left(h_{n,t,m}; \mu_{n,t,m}^{(i)}, \Gamma_{n,t,m}^{(i)}\right)\\
&\times \prod_{t^{\prime}\neq t}\delta_0(h_{n,t^{\prime},m})\Bigg),
\label{postdistri}
\end{aligned}
\end{align}


\noindent where $\mu_{n,t,m}^{(i)}=\frac{\beta_{n}r_{n,t,m}^{(i)}}{(\tau_{m}^{(i)})^{2}+\beta_{n}}$, $\Gamma_{n,t,m}^{(i)}=\frac{(\tau_{m}^{(i)})^{2}\beta_{n}}{(\tau_{m}^{(i)})^{2}+\beta_{n}}$, and the posterior sparsity ratio is given as
\begin{align}
\pi_{n,t,m}^{(i)}=\frac{e^{\frac{\beta_{n}|r_{n,t,m}^{(i)}|^2}{(\tau_{m}^{(i)})^{2}(\beta_{n}+(\tau_{m}^{(i)})^{2})}}}{\left(1\!+\!\frac{\beta_{n}}{(\tau_{m}^{(i)})^{2}}\right)\frac{1\!-\!\sum\limits_{t^{\prime}=0}^{T}\lambda_{n,t^{\prime},m}}{\lambda_{n,t,m}}\!+\!\sum\limits_{t^{\prime}=0}^{T}\frac{\lambda_{n,t^{\prime},m}}{\lambda_{n,t,m}}e^{\frac{\beta_{n}|r_{n,t^{\prime},m}^{(i)}|^2}{(\tau_{m}^{(i)})^{2}(\beta_{n}+(\tau_{m}^{(i)})^{2})}}}.
\label{postpara}
\end{align}
Therefore, the posterior mean (\ref{MMSEmean}) and variance (\ref{MMSEvar}) are given respectively as follows:
\begin{align}
[\hat{\eta}^{(i)}(\mathbf{r}_{n,m}^{(i)})]_t=\pi_{n,t,m}^{(i)}\mu_{n,t,m}^{(i)}.
\label{postmeancal}
\end{align}
\begin{align}
\psi_{n,t,m}^{(i)}=\pi_{n,t,m}^{(i)}(|\mu_{n,t,m}^{(i)}|^{2}+\Gamma_{n,t,m}^{(i)})-|[\hat{\eta}^{(i)}(\mathbf{r}_{n,m}^{(i)})]_{t}|^{2}.
\label{postvarcal}
\end{align}

\subsection{OAMP-based Algorithm for Multi-antenna BSs \label{extendtoMMV}}
The OAMP iteration for each BS antenna developed in Section \ref{SMV} neglects the common sparsity among received pilot signals of multiple BS antennas, which could be leveraged to enhance the estimation performance. Since the BS antennas receive pilot signal at the same instant, a common activity indicator is formed to update the prior information for all of them. Specifically, we replace $\lambda_{n,t,m}$ in (\ref{prior}) by variable $\omega_{n,t}^{\left(i\right)}$, which is updated according to the posterior sparsity ratio $\pi_{n,t,m}^{(i-1)}, \forall m$ obtained in the $(i-1)$-th OAMP iteration as follows:
\begin{align}
\omega_{n,t}^{(i)}=\frac{1}{M}\sum_{m=1}^{M}\pi_{n,t,m}^{(i-1)},
\label{sparsityratio}
\end{align}

\noindent i.e., the structured sparsity in the received signals of multiple BS antennas is learned by averaging their posterior sparsity ratios. Thus, the updated prior information in the $i$-th OAMP iteration for a multi-antenna BS is expressed as follows:
\begin{align}
\begin{aligned}
p\left(\mathbf{h}_{n,m}\right)&=\left(1-\sum_{t=0}^{T}\omega_{n,t}^{\left(i\right)}\right)\prod_{t=0}^{T}\delta_{0}\left(h_{n, t,m}\right)\\
&+\sum_{t=0}^{T}\left(\omega_{n,t}^{\left(i\right)}\mathcal{C N}\left(h_{n,t,m} ; 0, \beta_{n}\right)\prod_{t^{\prime}\neq t}\delta_0(h_{n,t^{\prime},m})\right).
\label{refinedprior}
\end{aligned}
\end{align}

\noindent Accordingly, we modify the posterior sparsity ratio as follows: 
\begin{align}
\pi_{n,t,m}^{(i)}=\frac{e^{\frac{\beta_{n}|r_{n,t,m}^{(i)}|^2}{(\tau_{m}^{(i)})^{2}(\beta_{n}+(\tau_{m}^{(i)})^{2})}}}{\left(1+\frac{\beta_{n}}{(\tau_{m}^{(i)})^{2}}\right)\frac{(1-\sum\limits_{t^{\prime}=0}^{T}\omega_{n,t^{\prime}}^{\left(i\right)})}{\omega_{n,t}^{\left(i\right)}}+\sum\limits_{t^{\prime}=0}^{T}\frac{\omega^{(i)}_{n,t^{\prime}}}{\omega^{(i)}_{n,t}}e^{\frac{\beta_{n}|r_{n,t^{\prime},m}^{(i)}|^2}{(\tau_{m}^{(i)})^{2}(\beta_{n}+(\tau_{m}^{(i)})^{2})}}}.
\label{refinedpostpara}
\end{align}

This fully formulates our proposed OAMP-based algorithm for joint activity-delay detection and channel estimation with a multi-antenna BS, with details summarized in Algorithm \ref{algo1}. After the OAMP-based algorithm terminates at the $I$-th iteration, the effective channel vector is estimated as $\hat{\mathbf{h}}_{m}=\hat{\eta}^{(I)}(\mathbf{r}_{m}^{(I)})$, $\forall m$, and the set of active users are determined as $\hat{\mathcal{K}}\triangleq \{n\in \mathcal{N}|\sum_{t=0}^{T} \omega_{n,t}^{(I+1)}\geq \theta\}$, where $\theta$ is an empirical threshold. Besides, the synchronization delay for an estimated active user is determined as $t_{n}= \mathop{\arg\max}\limits_{t \in \{0,\cdots,T\}} \omega_{n,t}^{(I+1)}$, $\forall n \in \hat{\mathcal{K}}$.
\begin{algorithm}[htpb]
\caption{The Proposed OAMP-based Algorithm \label{algo1}}
{\bf Input:}
The normalized received pilot signal $\mathbf{Y}$, pilot sequences $\mathbf{P}$, maximum number of iterations $Q$, maximum synchronization
delay $T$, and accuracy tolerance $\epsilon$.\\
{\bf Output:}
The estimated effective channel matrix $\hat{\mathbf{H}}$, set of active users $\hat{\mathcal{K}}$, and synchronization delays $\{t_{n}\}$'s, $n \in \hat{\mathcal{K}}$.\\
{\bf Initialize:}
$i \leftarrow 0$, $\mathbf{s}_{m}^{(1)}=\mathbf{0}$, $\forall m \in \mathcal{M}$, $v_{m}^{(1)}=0$, $\forall m\in \mathcal{M}$, $\omega_{n,t}^{(1)}=\frac{\lambda}{T+1}$, $\forall n \in \mathcal{N}$, $\forall t \in \{0,\cdots,T\}$.
\begin{algorithmic}[1]
\WHILE{{$i < Q$} \text{and} {$\frac{\sum_{m}||\mathbf{s}_{m}^{(i)}-\mathbf{s}_{m}^{(i-1)}||_{2}^{2}}{\sum_{m}||\mathbf{s}_{m}^{(i-1)}||_{2}^{2}} > \epsilon$}}
\STATE $i \leftarrow i+1$
\Statex \quad //\textit{The LE}//
\STATE Calculate $\mathbf{W}_{m}^{(i)}$, $\forall m$ according to (\ref{decor}).
\STATE Perform the LE to obtain $\mathbf{r}_{m}^{\left(i\right)}, \forall m$ via (\ref{LE}).
\STATE Calculate $(\tau_{m}^{(i)})^{2}$, $\forall m$ according to (\ref{errortau}).
\Statex \quad //\textit{The NLE}//
\STATE Update the prior sparsity ratio $\omega_{n,t}^{(i)}$, $\forall n, t$ and the
\Statex \quad corresponding prior information $p(\mathbf{h}_{n,m})$, $\forall n, m$ accord-
\Statex \quad ing to (\ref{sparsityratio}) and (\ref{refinedprior}), respectively.
\STATE Calculate the posterior mean $[\hat{\eta}^{(i)}(\mathbf{r}_{n,m}^{(i)})]_{t}$ and variance
\Statex \quad $\psi_{n,t,m}^{(i)}$, $\forall n, t, m$ according to (\ref{postmeancal}) and (\ref{postvarcal}), respectively.
\STATE Calculate $C_{m}^{(i)}$, $\frac{\sum_{n=1}^{N}\sum_{t=0}^{T}[\hat{\eta}^{\prime(i)}(\mathbf{r}_{n,m}^{(i)})]_{t}}{(T+1)N}$, and $(v_{m}^{(i+1)})^{2}$,
\Statex \quad $\forall n,t,m$ according to (\ref{Cm}), (\ref{bigterm}), and (\ref{v_calculate}), respectively.
\STATE Perform the NLE to obtain $\mathbf{s}_{m}^{(i+1)}$, $\forall m$ 
\Statex \quad according to (\ref{NLE}).
\ENDWHILE
\STATE Obtain $\hat{\mathbf{H}}=[\hat{\mathbf{h}}_{1},\cdots,\hat{\mathbf{h}}_{M}]$ with $\hat{\mathbf{h}}_{m}=\hat{\eta}^{(I)}(\mathbf{r}_{m}^{(I)})$.
\STATE Determine the set of active users as $\hat{\mathcal{K}}\triangleq \{n\in \mathcal{N}|\sum_{t=0}^{T} \omega_{n,t}^{(I+1)}\geq \theta\}$.
\STATE Determine the synchronization delays as $t_{n}=\mathop{\arg\max}\limits_{t \in \{0,\cdots,T\}} \omega_{n,t}^{(I+1)}$, $\forall n$.
\end{algorithmic}
\end{algorithm}

\section{A Low-Complexity FPAMP Approach for Asynchronous Massive RA \label{sectioniv}}
Although the OAMP-based algorithm is effective in solving the joint activity-delay detection and channel estimation problem, the high-complexity LMMSE estimator limits its application in practical mMTC systems that demand fast processing. In this section, we propose a novel algorithm for joint activity-delay detection and channel estimation based on the FPAMP framework, which has computational complexity at the same order as that of the standard AMP algorithm, while maintains similar performance as the OAMP-based algorithm.

\subsection{Free Probability Theory}
To develop the FPAMP-based algorithm, we first introduce basics of the free probability theory \cite{rspe2019,dvo1986}, which was established to calculate moments and distributions of non-commutative random variables, such as random matrices. For two independent scalar random variables $x_1$ and $x_2$, one can easily conclude $\mathbb{E}[x_1x_2x_1x_2]=\mathbb{E}[x_1^2x_2^2]$ with classical probability theory. However, the result is not necessarily valid for two random matrices $\mathbf{X}_{1}$ and $\mathbf{X}_{2}$ since they are non-commutative. Therefore, free probability theory brings in a concept of free independence between non-commutative random variables in analogous to the independence between commutative random variables. Since free independence is based on the non-commutative probability space, we review their definitions as follows.
\begin{definition}
(Non-commutative Probability Space) An algebraic non-commutative space is a pair ($\mathcal{A}$, $\phi$) consisting of a unital algebra $\mathcal{A}$, which consists an identity element $1$ (The unit of the algebra is usually denoted as $1$) with $1x = x = x1$, $\forall x \in \mathcal{A}$, and a linear function $\phi: \mathcal{A}\rightarrow \mathbb{C}$ such that $\phi(1)=1$. Elements from $\mathcal{A}\triangleq \{a_1,\cdots,a_k\}$ are addressed as non-commutative random variables.
\end{definition}

\begin{definition}
(Free Independence) Consider an algebraic non-commutative space ($\mathcal{A}$, $\phi$) and let $\{\mathcal{A}_{i}\}$ be the unital subalgebras of $\mathcal{A}$, i.e., $\mathcal{A}_{i} \subset \mathcal{A}$, $\forall i$. The subalgebras are said to be free independent if $\phi(a_1\cdots a_n)=0$ when $i_1 \neq i_2$, $i_2 \neq i_3$, $\cdots$, $i_{n-1} \neq i_{n}$, $a_k \in \mathcal{A}_{i_{k}}$ and $\phi(a_k)=0$, $\forall k = 1,\cdots, n$. The random variables are called free independent if their generated unital subalgebras are free independent.
\end{definition}

While it is difficult to explicitly define the PDFs for non-commutative random variables (e.g., random matrices), we turn to their moments since the characteristic functions and PDFs can often be obtained based on moments. The free cumulants of symmetric random matrices for moment calculations are defined below.
\begin{definition}
(Free Cumulants) Let $X$ be a random variable with finite moments of all orders and $m_{k}\triangleq \mathbb{E}[X^{k}]$ be its $k$-th moment. Suppose the law of $X$ is the empirical eigenvalue distribution of a symmetric random matrix $\mathbf{A} \in \mathbb{C}^{n\times n}$. The free cumulants of $\mathbf{A}$, denoted by $(\kappa_{k})_{k=1}^{\infty}$, are defined recursively via the moment-cumulant relation as follows:
\begin{align}
m_{k}= \sum_{\pi \in \mathcal{NC}(k)} \prod_{j=1}^{k} \kappa_{j}^{N_{j}(\pi)},
\end{align}
where $\pi$ is a partition of $\mathcal{NC}(k)$, i.e., the set of all non-crossing partitions of $[k] \triangleq \{1,2,\cdots,k\}$\footnote{A partition $\pi$ of set $[k]$ is defined as $\pi \triangleq \{V_1,\cdots,V_s\}$ such that $V_1,\cdots,V_s\subset [k]$, where $V_i \neq \emptyset$, $V_i \cap V_j = \emptyset$, $i\neq j \in \{1,\cdots, k\}$, and $V_1 \cup \cdots \cup V_s = [k]$. Subsets $V_1,\cdots,V_s$ are called the blocks of $\pi$, and the length of a block denotes the number of its elements. Besides, if there exists $i<j<l<s$ such that $i$ and $l$ are in one block $V_p$, and $j$ and $s$ are in another block $V_q$, we call that $V_p$ and $V_q$ cross. If there is no pair of blocks in $\pi$ crosses, $\pi$ is called a non-crossing partition.}, and $N^{j}(\pi)$ represents the number of length-$j$ blocks in $\pi$.
\end{definition}

As will be seen later, the empirical eigenvalue distributions of random matrices are critical to the FPAMP framework. However, since the pilot matrix $\mathbf{P}$ is not square, we next introduce the rectangular free cumulants of rectangular random matrices.
\begin{definition}
(Rectangular Free Cumulants) Let $X$ be a random variable with finite moments of all orders and $m_{2k}\triangleq \mathbb{E}[X^{2k}]$ be its even moments. Suppose the law of $X^2$ represents the empirical eigenvalue distribution of $\mathbf{A}\mathbf{A}^{H}$ with $\mathbf{A} \in \mathbb{C}^{n\times d}$ denoting a rectangular random matrix. The rectangular free cumulants $\left\{\kappa_{2 k}\right\}_{k \geq 1}$ of $\mathbf{A}$ are defined recursively via the moment-cumulant relation as follows:
\begin{align}
m_{2 k}= \frac{n}{d}\sum_{\pi \in \mathcal{NC}(2 k)} \prod_{\substack{\mathcal{S} \in \pi \\ \min \mathcal{S} \text { is odd }}} \kappa_{|\mathcal{S}|} \prod_{\substack{\mathcal{S}^{\prime} \in \pi \\ \min \mathcal{S}^{\prime} \text { is even }}} \kappa_{|\mathcal{S}^{\prime}|},
\end{align}
\noindent where each $\mathcal{S} \in \pi \in \mathcal{NC}\left(2k\right)$ has even cardinality.
\end{definition}
With the basic knowledge on free probability theory, we propose the low-complexity FPAMP-based receiver in the next subsection.

\subsection{The Proposed FPAMP-based Algorithm}
We propose to reduce the high computational complexity caused by LMMSE estimator in the OAMP-based algorithm by utilizing the memorized information from all the preceding iterations. This is motivated by a solution to the Thouless-Anderson-Palmer (TAP) equations of Ising models with general invariant matrices \cite{mopper2016,zfan2020}, where iterative algorithms that utilize the information of preceding iterations are analyzed. Our proposed FPAMP-based algorithm evolves from the AMP framework, which takes a major detour from the memory AMP algorithm \cite{xbian2023conference} that was developed under the OAMP framework. Based on the normalized received pilot signal for individual BS antennas in (\ref{SMVmodel}), the FPAMP-based algorithm iteratively performs the following updates:
\begin{align}
\mathbf{h}_{m}^{(i)} = \mathbf{P}^{\mathrm{H}}\mathbf{s}_{m}^{(i)}-\sum_{j=1}^{i-1}\gamma_{m,j}^{(i)}\tilde{\mathbf{h}}_{m}^{(j)},
\label{fpamp1}
\end{align}
\begin{align}
\tilde{\mathbf{h}}_{m}^{(i)} = f^{(i)}(\mathbf{h}_{m}^{(1)},\cdots,\mathbf{h}_{m}^{(i)}),
\label{fpamp2}
\end{align}
\begin{align}
\mathbf{r}_{m}^{(i)} = \mathbf{P}\tilde{\mathbf{h}}_{m}^{(i)}-\sum_{j=1}^{i}\alpha_{m,j}^{(i)}\mathbf{s}_{m}^{(j)},
\label{fpamp3}
\end{align}
\begin{align}
\mathbf{s}_{m}^{(i+1)} = g^{(i+1)}(\mathbf{r}_{m}^{(1)},\cdots,\mathbf{r}_{m}^{(i)},\mathbf{y}_{m}),
\label{fpamp4}
\end{align}
where $\mathbf{s}_{m}^{(1)}=\mathbf{y}_{m}$ and $\mathbf{h}_{m}^{(1)}=\mathbf{P}^{\mathrm{H}}\mathbf{s}_{m}^{(1)}$. In (\ref{fpamp2}) and (\ref{fpamp4}), $f^{(i)}\left(\cdot\right)$ and $g^{(i)}\left(\cdot\right)$ are Lipschitz functions, which are called denoisers and will be developed in the following. Besides, $\{\alpha_{m,j}^{(i)}\}_{j=1}^{i}$ and $\{\gamma_{m,j}^{(i)}\}_{j=1}^{i-1}$ are the Onsager parameters, which restrict the joint distributions of $\{\mathbf{h}_{m}^{(i)},\mathbf{r}_{m}^{(i)}\}$ conditioned on $\mathbf{h}_{m}$ to Gaussian. 

\subsubsection{Onsager coefficients} We compute the Onsager coefficients via auxiliary matrices $\mathbf{\Psi}_{m}^{(i+1)} \in\mathbb{C}^{\left(i+1\right)\times\left(i+1\right)}$, $\mathbf{\Phi}_{m}^{(i+1)} \in \mathbb{C}^{(i+1)\times(i+1)}$ defined as follows:
\begin{align}
\mathbf{\Psi}_{m}^{(i+1)}=\left[\begin{array}{ccccc}
0 & 0 & \ldots & 0 & 0 \\
0 & \left\langle\partial_1 \tilde{\mathbf{h}}_{m}^{(1)}\right\rangle & 0 & \ldots & 0 \\
0 & \left\langle\partial_1 \tilde{\mathbf{h}}_{m}^{(2)}\right\rangle & \left\langle\partial_2 \tilde{\mathbf{h}}_{m}^{(2)}\right\rangle & \ldots & 0 \\
\vdots & \vdots & \vdots & \ddots & \vdots \\
0 & \left\langle\partial_1 \tilde{\mathbf{h}}_{m}^{(i)}\right\rangle & \left\langle\partial_2 \tilde{\mathbf{h}}_{m}^{(i)}\right\rangle & \ldots & \left\langle\partial_i \tilde{\mathbf{h}}_{m}^{(i)}\right\rangle
\end{array}\right],
\end{align}
\begin{align}
\mathbf{\Phi}_{m}^{(i+1)}=\left[\begin{array}{ccccc}
0 & 0 & \ldots & 0 & 0 \\
\left\langle\partial_{\mathbf{G}_{m}} \mathbf{s}_{m}^{(1)}\right\rangle & 0 & 0 & \ldots & 0 \\
\left\langle\partial_{\mathbf{G}_{m}} \mathbf{s}_{m}^{(2)}\right\rangle & \left\langle\partial_1 \mathbf{s}_{m}^{(2)}\right\rangle & 0 & \ldots & 0 \\
\vdots & \vdots & \ddots & \vdots & \vdots \\
\left\langle\partial_{\mathbf{G}_{m}} \mathbf{s}_{m}^{(i)}\right\rangle & \left\langle\partial_1 \mathbf{s}_{m}^{(i)}\right\rangle & \ldots & \left\langle\partial_{i-1} \mathbf{s}_{m}^{(i)}\right\rangle & 0
\end{array}\right],
\end{align}
where $\partial_k \tilde{\mathbf{h}}_{m}^{(i)} \in \mathbb{C}^{(T+1)N\times 1}$ denotes the partial derivatives $\Big[\partial_{h_{1,m}^{(k)}} f^{(i)}(h_{1,m}^{(1)},\cdots,h_{1,m}^{(i)}), \cdots,\partial_{h_{(T+1) N,m}^{(k)}}f^{(i)}(h_{(T+1)N,m}^{(1)},\!\cdots\!,$ $h_{(T+1) N,m}^{(i)})\Big]^{\mathrm{T}}$ for $k \in \{1,\cdots, i\}$ with $h_{l,m}^{(i)}$ denoting the $l$-th element of $\mathbf{h}_{m}^{(i)}$. Besides, $\mathbf{G}_{m} \triangleq \mathbf{P}\mathbf{h}_{m}$ and $\partial_{\mathbf{G}_{m}} \mathbf{s}_{m}^{(i)} \in \mathbb{C}^{L\times 1}$ denotes the partial derivatives $\Big[\partial_{G_{1,m}} g^{(i)}(r_{1,m}^{(1)},\cdots,r_{1,m}^{(i-1)}, y_{1,m}),\cdots,\partial_{G_{L,m}} g^{(i)}(r_{L,m}^{(1)},\cdots,$ $r_{L,m}^{(i-1)},y_{L,m})\Big]^{\mathrm{T}}$ with $G_{l,m}$, $r_{l,m}^{(i)}$, and $y_{l,m}$ denoting the $l$-th element of $\mathbf{G}_{m}$, $\mathbf{r}_{m}^{(i)}$, and $\mathbf{y}_{m}$, respectively. Similarly, $\partial_k \mathbf{s}_{m}^{(i)} \in \mathbb{C}^{L\times 1}$ denotes the partial derivatives $\Big[\partial_{r_{1,m}^{(k)}} g^{(i)}(r_{1,m}^{(1)},\cdots,r_{1,m}^{(i-1)},y_{1,m}),\cdots,\partial_{r_{L,m}^{(k)}} g^{(i)}(r_{L,m}^{(1)},\cdots,$ $r_{L,m}^{(i-1)},y_{L,m})\Big]^{\mathrm{T}}$ for $k \in \{1,\cdots, i-1\}$. 

Then, we define matrices $\mathbf{M}_{m,\alpha}^{(i+1)}$, $\mathbf{M}_{m,\gamma}^{(i+1)} \in \mathbb{C}^{(i+1)\times (i+1)}$ as follows:
\begin{align}
\mathbf{M}^{(i+1)}_{m,\alpha}=\sum_{j=0}^{i+1} \kappa_{2(j+1)} \mathbf{\Psi}_{m}^{(i+1)}\left(\mathbf{\Phi}_{m}^{(i+1)} \mathbf{\Psi}_{m}^{(i+1)}\right)^j,
\label{Malpha}
\end{align}
\begin{align}
\mathbf{M}^{(i+1)}_{m,\gamma}=\frac{L}{(T+1)N} \sum_{j=0}^{i} \kappa_{2(j+1)} \mathbf{\Phi}_{m}^{(i+1)}\left(\mathbf{\Psi}_{m}^{(i+1)} \mathbf{\Phi}_{m}^{(i+1)}\right)^j,
\label{Mgamma}
\end{align}
where $\{\kappa_{2k}\}_{k \geq 1}$ denotes the rectangular free cumulants of $\mathbf{P}$, which can be calculated by exploiting the connection with the formal power series, i.e., a generalization of polynomials that allows an infinite number of terms. In particular, $\kappa_{2}= m_{2}$, and when $k>1$ \cite{fl2009},
\begin{align}
\kappa_{2 k}\!=\! m_{2 k}\!-\!\left[z^k\right]\!\left(\sum_{j=1}^{k-1} \!\kappa_{2 j}\!\left(z\!\left(\frac{L}{(T+1)N} M(z)+1\!\right)\!\left(M(z)+1\right)\right)^j\!\right).
\label{freecumulant}
\end{align}
Note that $\{m_{2 k}\}_{k=1}^{i}$ denotes the first $i$ moments of the empirical eigenvalue distribution of $\mathbf{P}\mathbf{P}^{\mathrm{H}}$, $M(z)\triangleq \sum_{k^{\prime}=1}^{\infty} m_{2 k^{\prime}} z^{k^{\prime}}$, and $\left[z^k\right](q(z))$ denotes the coefficient of $z^k$ in polynomial $q(z)$. The moments $\{m_{2 k}\}_{k=1}^{i}$ can be estimated in $\mathcal{O}(L(T+1)N)$ following \cite{marco2022}, which is at the same complexity order of the proposed FPAMP-based algorithm in one iteration as will be analyzed in Section \ref{sectioncomplexity}. To do so, one may sample a standard complex Gaussian vector $\mathbf{c}_{0} \sim \mathcal{CN}(\mathbf{0}_{L},\mathbf{I}_{L})$ and compute $\mathbf{c}_{k}=\mathbf{P}^{\mathrm{H}}\mathbf{c}_{k-1}$ for odd $k$'s and $\mathbf{c}_{k}=\mathbf{P}\mathbf{c}_{k-1}$ for even $k$'s. Thus, $\frac{||\mathbf{c}_{k}||^{2}}{(T+1)L}$ is a consistent estimate of the $k$-th moment of the empirical eigenvalue distribution of $\mathbf{P}\mathbf{P}^{\mathrm{H}}$. Besides, we assume that the empirical distribution of the singular values of $\mathbf{P}$ converges weakly almost surely to a random variable with its even moments and rectangular free cumulants expressed as $\{\bar{m}_{2k}\}_{k\geq1}$ and $\{\bar{\kappa}_{2k}\}_{k\geq 1}$. Therefore, as $L,N \rightarrow \infty$, $m_{2k} \rightarrow \bar{m}_{2k}$ and $\kappa_{2k} \rightarrow \bar{\kappa}_{2k}$, $\forall k$.

Finally, we obtain the Onsager coefficients $\{\alpha_{m,j}^{(i)}\}_{j=1}^{i}$ and $\{\gamma_{m,j}^{(i)}\}_{j=1}^{i-1}$ from the last row of $\mathbf{M}^{(i+1)}_{m,\alpha}$ and $\mathbf{M}^{(i+1)}_{m,\gamma}$ as follows:
\begin{align}
\left(\alpha_{m,1}^{(i)}, \ldots, \alpha_{m,i}^{(i)}\right)=\left(\left[\mathbf{M}_{m,\alpha}^{(i+1)}\right]_{i+1,2}, \ldots,\left[\mathbf{M}_{m,\alpha}^{(i+1)}\right]_{i+1, i+1}\right),
\label{alphaupdate}
\end{align}
\begin{align}
\left(\gamma_{m,1}^{(i)}, \ldots, \gamma_{m, i-1}^{(i)}\right)=\left(\left[\mathbf{M}_{m,\gamma}^{(i+1)}\right]_{i+1,2}, \ldots,\left[\mathbf{M}_{m,\gamma}^{(i+1)}\right]_{i+1, i}\right).
\label{gammaupdate}
\end{align}

\subsubsection{State evolution and joint empirical distributions} In order to design the denoisers $f^{(i)}\left(\cdot\right)$ and $g^{(i)}\left(\cdot\right)$, we first present the state evolution of the FPAMP-based algorithm and the convergence of the joint empirical distributions of related random variables, as shown in Theorem \ref{theoremconverge}. To proceed, we introduce the definition of $W_2$-convergence \cite{cvi2009}.
\begin{definition}
    ($W_2$-convergence) We use $(\mathbf{g}^{(1)},\cdots,\mathbf{g}^{(i)}) \stackrel{W_2}{\longrightarrow}(G^{(1)}, \cdots, G^{(i)})$ to denote the joint empirical distributions of the rows of the random matrix $[\mathbf{g}^{(1)},\cdots,\mathbf{g}^{(i)}] \in \mathbb{C}^{d \times i}$ converge to the law of random vector $[G^{(1)}, \cdots, G^{(i)}]$ in the Wasserstein space of order 2. Specifically [Corollary 7.4, 36], if a function $\psi$: $\mathbb{C}^{i}\rightarrow \mathbb{C}$ satisfies
    \begin{align}
        |\psi(\mathbf{u})-\psi(\mathbf{v})| \leq L\|\mathbf{u}-\mathbf{v}\|_2 (1+\|\mathbf{u}\|_2+\|\mathbf{v}\|_2)
    \end{align}
    for all $\mathbf{u}$, $\mathbf{v} \in \mathbb{C}^{i}$, and $L>0$, we have
    \begin{align}
        \lim _{d \rightarrow \infty} \frac{1}{d} \sum_{j=1}^d \psi\left(g_j^{(1)}, \ldots, g_j^{(i)}\right)=\mathbb{E}\left\{\psi\left(G^{(1)}, \ldots, G^{(i)}\right)\right\},
    \end{align}
    if $(\mathbf{g}^{(1)},\cdots,\mathbf{g}^{(i)}) \stackrel{W_2}{\longrightarrow}(G^{(1)}, \cdots, G^{(i)})$. In addition, function $\psi$ is pseudo-Lipschitz of order 2.
\end{definition}
\begin{theorem}
\label{theoremconverge}
    Let $\psi$: $\mathbb{C}^{2i+1}\rightarrow \mathbb{C}$ and $\phi$: $\mathbb{C}^{2i+2}\rightarrow \mathbb{C}$ be any pseudo-Lipschitz functions of order 2. For $i = 1, 2, \cdots$, it is almost sure that
    \begin{align}
    \begin{aligned}
    \label{theoremconverge1}
        \lim _{(T+1)N \rightarrow \infty} &\frac{\sum_{t=1}^{(T+1)N} \psi\left(h_{t,m}^{(1)}, \ldots, h_{t,m}^{(i)}, \tilde{h}_{t,m}^{(1)}, \ldots, \tilde{h}_{t,m}^{(i)}, h_{t,m}\right)}{(T+1)N}\\
        &=\mathbb{E}\left[\psi\left(H_{m}^{(1)}, \ldots, H_{m}^{(i)}, \tilde{H}_{m}^{(1)}, \ldots, \tilde{H}_{m}^{(i)}, H_{m}\right)\right],
    \end{aligned}
    \end{align}
    \begin{align}
    \begin{aligned}
    \label{theoremconverge2}
        \lim _{L \rightarrow \infty} &\frac{1}{L} \sum_{l=1}^{L} \phi\left(r_{l,m}^{(1)}, \ldots, r_{l,m}^{(i)}, s_{l,m}^{(1)}, \ldots, s_{l,m}^{(i+1)}, y_{l,m}\right)\\
        &=\mathbb{E}\left[\phi\left(R_{m}^{(1)}, \ldots, R_{m}^{(i)}, S_{m}^{(1)}, \ldots, S_{m}^{(i+1)}, Y_{m}\right)\right],
    \end{aligned}
    \end{align}
    where $H_{m}$, $H^{(i)}_{m}$, $\tilde{H}^{(i)}_{m}$, $S^{(i)}_{m}$, $R^{(i)}_{m}$, and $Y_{m}$ are defined in the state evolution of the FPAMP-based algorithm as follows:
    \begin{align}
    \left[G_{m},R_{m}^{(1)}, \ldots, R_{m}^{(i-1)}\right]^{\mathrm{T}} \sim \mathcal{CN}(\mathbf{0},\bar{\mathbf{\Sigma}}_{m}^{(i)}),
    \label{SE1}
    \end{align}
    \begin{align}
    S_{m}^{(i)}=g^{(i)}(R_{m}^{(1)}, \ldots, R_{m}^{(i-1)},Y_{m}),
    \label{SE2}
    \end{align}
    \begin{align}   \left[H_{m}^{(1)},\cdots,H_{m}^{(i)}\right]^{\mathrm{T}}=\bar{\boldsymbol{\mu}}_{m}^{(i)}H_{m} +\left[W_{m}^{(1)},\cdots,W_{m}^{(i)}\right]^{\mathrm{T}},
    \label{SE3}
    \end{align}
    \begin{align}
        \tilde{H}_{m}^{(i)}=f^{(i)}(H_{m}^{(1)},\cdots,H_{m}^{(i)}).
    \label{SE4}
    \end{align}
    In particular, $\mathbf{h}_{m} \stackrel{W_2}{\longrightarrow} H_m$, $G_m \sim \mathcal{CN}(0,\bar{\kappa}_2\mathbb{E}[H_{m}^2])$, and $Y_m=G_m+N_m$ with $\mathbf{n}_{m} \stackrel{W_2}{\longrightarrow} N_m$. Besides, $\bar{\boldsymbol{\mu}}_{m}^{(i)}=[\bar{\mu}_{m}^{(1)},\cdots,\bar{\mu}_{m}^{(i)}]^{\mathrm{T}}$ and $\Big[W_{m}^{(1)},\cdots,W_{m}^{(i)}\Big] \sim \mathcal{CN}(\mathbf{0},\bar{\mathbf{\Omega}}_{m}^{(i)})$, which are independent of $H_{m}$. Equivalently, as $L,N \rightarrow \infty$, the joint empirical distributions of $(\mathbf{h}_{m}^{(1)},\cdots,\mathbf{h}_{m}^{(i)},\tilde{\mathbf{h}}_{m}^{(1)},\cdots,\tilde{\mathbf{h}}_{m}^{(i)},$ $\mathbf{h}_{m})$ and $(\mathbf{r}_{m}^{(1)},\cdots,\mathbf{r}_{m}^{(i)}$, $\mathbf{s}_{m}^{(1)},\cdots,\mathbf{s}_{m}^{(i+1)},\mathbf{y}_{m})$ converge almost surely in Wasserstein-2 distance to $(H_{m}^{(1)}, \ldots,$ $ H_{m}^{(i)}, \tilde{H}_{m}^{(1)}, \ldots,\tilde{H}_{m}^{(i)}, H_{m})$ and $(R_{m}^{(1)}, \ldots, R_{m}^{(i)}, S_{m}^{(1)}, \ldots, S_{m}^{(i+1)}, Y_{m})$, respectively.
\end{theorem}
\begin{proof}[Proof Sketch]
    Due to the limited space, we highlight the key idea of the proof herein, which is based on the general recursion and its state evolution, which characterizes the common structure of the FPAMP-based algorithm. We first prove that the joint empirical distributions of some variables in the general recursion converge to the multivariate Gaussian distribution in its state evolution. Then, it can be shown via an induction that the state evolution parameters of the general recursion match those of the FPAMP-based algorithm, and the FPAMP iterates are close to the general recursion. Therefore, the joint empirical distributions of some variables in the FPAMP-based algorithm also converge to the multivariate Gaussian distribution in its state evolution. Detailed derivations are provided in Appendix \ref{appendixa}.
\end{proof}

\subsubsection{Denoisers} Based on the state evolution of the FPAMP-based algorithm in (\ref{SE1})$\sim$(\ref{SE4}), $\Big[(\mathbf{h}_{m}^{(1)}-\bar{\mu}_{m}^{(1)}\mathbf{h}_{m})^{\mathrm{T}},\!\cdots\!,$ $(\mathbf{h}_{m}^{(i)}-\bar{\mu}_{m}^{(i)}\mathbf{h}_{m})^{\mathrm{T}}\Big]^{\mathrm{T}}$ and $\Big[\mathbf{G}_{m}^{\mathrm{T}},(\mathbf{r}_{m}^{(1)})^{\mathrm{T}},\!\cdots,\!(\mathbf{r}_{m}^{(i)})^{\mathrm{T}}\Big]^{\mathrm{T}}$ converge in the large system limit, i.e., $L,N \rightarrow \infty$, to two independent zero-mean multi-variate complex Gaussian distributions with covariance matrices $\bar{\mathbf{\Omega}}_{m}^{(i)}$ and $\bar{\mathbf{\Sigma}}_{m}^{(i+1)}$, respectively \cite{zfan2020,marco2022}. This property can be used to design the denoisers, with which, we only need to track the mean vector $\bar{\boldsymbol{\mu}}_{m}^{(i)}$ and covariance matrices $\bar{\mathbf{\Omega}}_{m}^{(i)}$, $\bar{\mathbf{\Sigma}}_{m}^{(i+1)}$.

First, we define auxiliary matrices $\bar{\mathbf{\Psi}}_{m}^{(i+1)}$, $\bar{\mathbf{\Phi}}_{m}^{(i+1)}$, $\bar{\mathbf{\Gamma}}_{m}^{(i+1)}$, $\bar{\mathbf{\Delta}}_{m}^{(i+1)}$ as follows:
\begin{align}
\bar{\mathbf{\Psi}}_{m}^{(i+1)}\!=\!\left[\begin{array}{ccccc}
0 & 0 & \ldots & 0 & 0 \\
0 & \mathbb{E}[\partial_1 \tilde{H}^{(1)}_{m}] & 0 & \ldots & 0 \\
0 & \mathbb{E}[\partial_1 \tilde{H}^{(2)}_{m}] & \mathbb{E}[\partial_2 \tilde{H}^{(2)}_{m}] & \ldots & 0 \\
\vdots & \vdots & \vdots & \ddots & \vdots \\
0 & \mathbb{E}[\partial_1 \tilde{H}^{(i)}_{m}] & \mathbb{E}[\partial_2 \tilde{H}^{(i)}_{m}] & \ldots & \mathbb{E}[\partial_i \tilde{H}^{(i)}_{m}]
\end{array}\right],
\label{psibar}
\end{align}
\begin{align}
\bar{\mathbf{\Phi}}_{m}^{(i+1)}\!=\!\left[\begin{array}{ccccc}
0 & 0 & \ldots & 0 & 0 \\
\mathbb{E}[\partial_{G_{m}} S^{(1)}_{m}] & 0 & 0 & \ldots & 0 \\
\mathbb{E}[\partial_{G_{m}} S^{(2)}_{m}] & \mathbb{E}[\partial_{1} S^{(2)}_{m}] & 0 & \ldots & 0 \\
\vdots & \vdots & \ddots & \vdots & \vdots \\
\mathbb{E}[\partial_{G_{m}} S^{(i)}_{m}] & \mathbb{E}[\partial_{1} S^{(i)}_{m}] & \ldots & \mathbb{E}[\partial_{i-1} S^{(i)}_{m}] & 0
\end{array}\right],
\label{phibar}
\end{align}
\begin{align}
\bar{\mathbf{\Gamma}}_{m}^{(i+1)}\!=\!\left[\begin{array}{cccc}
\mathbb{E}[H_{m}^2] & \mathbb{E}[H_{m}\tilde{H}_{m}^{(1)}] & \ldots & \mathbb{E}[H_{m}\tilde{H}_{m}^{(i)}] \\
\mathbb{E}[H_{m}\tilde{H}_{m}^{(1)}] & \mathbb{E}[(\tilde{H}_{m}^{(1)})^2] & \ldots & \mathbb{E}[\tilde{H}_{m}^{(1)}\tilde{H}_{m}^{(i)}] \\
\mathbb{E}[H_{m}\tilde{H}_{m}^{(2)}] & \mathbb{E}[\tilde{H}_{m}^{(1)}\tilde{H}_{m}^{(2)}] & \ldots & \mathbb{E}[\tilde{H}_{m}^{(2)}\tilde{H}_{m}^{(i)}] \\
\vdots & \vdots & \ddots & \vdots \\
\mathbb{E}[H_{m}\tilde{H}_{m}^{(i)}] & \mathbb{E}[\tilde{H}_{m}^{(1)}\tilde{H}_{m}^{(i)}]  & \ldots & \mathbb{E}[(\tilde{H}_{m}^{(i)})^2]
\end{array}\right],
\label{Gammabar}
\end{align}
\begin{align}
\bar{\mathbf{\Delta}}_{m}^{(i+1)}\!=\!\left[\begin{array}{ccccc}
0 & 0 & \ldots & 0 & 0 \\
0 & \mathbb{E}[(S_{m}^{(1)})^2] & \mathbb{E}[S_{m}^{(1)}S_{m}^{(2)}] & \ldots & \mathbb{E}[S_{m}^{(1)}S_{m}^{(i)}] \\
0 & \mathbb{E}[S_{m}^{(1)}S_{m}^{(2)}] & \mathbb{E}[(S_{m}^{(2)})^2] & \ldots & \mathbb{E}[S_{m}^{(2)}S_{m}^{(i)}] \\
\vdots & \vdots & \vdots & \ddots & \vdots \\
0 & \mathbb{E}[S_{m}^{(1)}S_{m}^{(i)}] & \mathbb{E}[S_{m}^{(2)}S_{m}^{(i)}] & \ldots & \mathbb{E}[(S_{m}^{(i)})^2]
\end{array}\right],
\label{Deltabar}
\end{align}
where $\partial_k \tilde{H}_{m}^{(i)}$, $\partial_{G_{m}} S_{m}^{(i)}$, and $\partial_{k} S_{m}^{(i)}$ denote the partial derivatives $\partial_{H_{m}^{(k)}} f^{(i)}(H_{m}^{(1)},\cdots,H_{m}^{(i)})$, $\partial_{G_{m}}g^{(i)}$ $(R_{m}^{(1)},\cdots,R_{m}^{(i-1)},Y_{m})$, and $\partial_{R_{m}^{(k)}}g^{(i)}(R_{m}^{(1)},\cdots,R_{m}^{(i-1)},Y_{m})$. Starting with $\bar{\mathbf{\Sigma}}_{m}^{(1)}\!=\!\bar{\kappa}_2\mathbb{E}[H_{m}^2]$, $\bar{\boldsymbol{\mu}}_{m}^{(1)}=\frac{L\bar{\kappa}_2\mathbb{E}[\partial_{G_{m}} g^{(1)}(Y_{m})]}{(T+1)N}$, and $\bar{\mathbf{\Omega}}_{m}^{(1)}=\frac{L\left(\bar{\kappa}_2\mathbb{E}[(g^{(1)}(Y_{m}))^2]+\bar{\kappa}_4\mathbb{E}[H_{m}^2](\mathbb{E}[\partial_{G_{m}}g^{(1)}(Y_{m})])^2\right)}{(T+1)N}$, the covariance matrix $\bar{\mathbf{\Sigma}}_{m}^{(i+1)}$ is updated as follows:
\begin{align}
\bar{\mathbf{\Sigma}}_{m}^{(i+1)} = \sum_{j=0}^{2i+1}\bar{\kappa}_{2(j+1)}\mathbf{\Xi}_{m,j}^{(i+1)},
\label{Sigmaupdate}
\end{align}
where $\mathbf{\Xi}_{m,0}^{(i+1)}=\bar{\mathbf{\Gamma}}_{m}^{(i+1)}$, and for $j \geq 1$,
\begin{align}
\begin{aligned}
\mathbf{\Xi}_{m,j}^{(i+1)} & =\sum_{k=0}^j\left(\bar{\mathbf{\Psi}}_{m}^{(i+1)} \bar{\mathbf{\Phi}}_{m}^{(i+1)}\right)^{k} \bar{\mathbf{\Gamma}}_{m}^{(i+1)}\left(\left(\bar{\mathbf{\Psi}}_{m}^{(i+1)} \bar{\mathbf{\Phi}}_{m}^{(i+1)}\right)^{\mathrm{H}}\right)^{j-k} \\
& +\sum_{k=0}^{j-1}\left(\bar{\mathbf{\Psi}}_{m}^{(i+1)} \bar{\mathbf{\Phi}}_{m}^{(i+1)}\right)^k \bar{\mathbf{\Psi}}_{m}^{(i+1)} \bar{\mathbf{\Delta}}_{m}^{(i+1)} (\bar{\mathbf{\Psi}}_{m}^{(i+1)})^{\mathrm{H}}\\
&\quad \quad \cdot\left(\left(\bar{\mathbf{\Psi}}_{m}^{(i+1)} \bar{\mathbf{\Phi}}_{m}^{(i+1)}\right)^{\mathrm{H}}\right)^{j-k-1}.
\end{aligned}
\label{Sigmaupdate1}
\end{align}
Next, we define a symmetric matrix $\check{\mathbf{\Omega}}_{m}^{(i+2)} \in \mathbb{C}^{(i+2) \times (i+2)}$ by padding one row and one column of zeros before the first row and first column of $\bar{\mathbf{\Omega}}_{m}^{(i+1)} \in \mathbb{C}^{(i+1) \times (i+1)}$. In particular, we can compute $\check{\mathbf{\Omega}}_{m}^{(i+2)}$ as follows:
\begin{align}
    \check{\mathbf{\Omega}}_{m}^{(i+2)} = \frac{L}{(T+1)N}\sum_{j=0}^{2(i+1)}\bar{\kappa}_{2(j+1)}\mathbf{\Theta}_{m,j}^{(i+2)},
\label{Omegaupdate1}
\end{align}
where $\mathbf{\Theta}_{m,0}^{(i+2)}=\bar{\mathbf{\Delta}}_{m}^{(i+2)}$, and for $j\geq 1$,
\begin{align}
\begin{aligned}
\mathbf{\Theta}_{m,j}^{(i+2)}  &=\sum_{k=0}^j\left(\bar{\mathbf{\Phi}}_{m}^{(i+2)}\bar{\mathbf{\Psi}}_{m}^{(i+2)} \right)^{k} \bar{\mathbf{\Delta}}_{m}^{(i+2)}\left(\left(\bar{\mathbf{\Phi}}_{m}^{(i+2)}\bar{\mathbf{\Psi}}_{m}^{(i+2)} \right)^{\mathrm{H}}\right)^{j-k} \\
& +\sum_{k=0}^{j-1}\left(\bar{\mathbf{\Phi}}_{m}^{(i+2)}\bar{\mathbf{\Psi}}_{m}^{(i+2)} \right)^k \bar{\mathbf{\Phi}}_{m}^{(i+2)} \bar{\mathbf{\Gamma}}_{m}^{(i+2)} (\bar{\mathbf{\Phi}}_{m}^{(i+2)})^{\mathrm{H}}\\
&\quad \quad \cdot \left(\left(\bar{\mathbf{\Phi}}_{m}^{(i+2)}\bar{\mathbf{\Psi}}_{m}^{(i+2)} \right)^{\mathrm{H}}\right)^{j-k-1}.
\label{Omegaupdate2}
\end{aligned}
\end{align}
Therefore, $\bar{\mathbf{\Omega}}_{m}^{(i+1)}$ is expressed as follows:
\begin{align}
    \bar{\mathbf{\Omega}}_{m}^{(i+1)} = [\check{\mathbf{\Omega}}_{m}^{(i+2)}]_{2:i+2,2:i+2}.
\label{Omegaupdate3}
\end{align}
Finally, we evaluate the mean vector $\bar{\boldsymbol{\mu}}_{m}^{(i+1)}$ recursively with each element given as follows: 
\begin{align}
    \bar{\mu}_{m}^{(i+1)}=[\bar{\mathbf{M}}^{(i+2)}_{m,\gamma}]_{i+2,1},
\label{muupdate}
\end{align}
where $\bar{\mathbf{M}}^{(i+2)}_{m,\gamma}$ is obtained as
\begin{align}
\bar{\mathbf{M}}^{(i+2)}_{m,\gamma}=\frac{L}{(T+1)N} \sum_{j=0}^{i+1} \bar{\kappa}_{2(j+1)} \bar{\mathbf{\Phi}}_{m}^{(i+2)}\left(\bar{\mathbf{\Psi}}_{m}^{(i+2)} \bar{\mathbf{\Phi}}_{m}^{(i+2)}\right)^j.
\label{Mbarupdate}
\end{align}

Note that the formulas of $\mathbf{\Theta}_{m,j}^{(i+2)}$ and $\bar{\mathbf{M}}^{(i+2)}_{m,\gamma}$ in (\ref{Omegaupdate2}) and (\ref{Mbarupdate}) involve matrices $\bar{\mathbf{\Psi}}_{m}^{(i+2)}$ and $\bar{\mathbf{\Gamma}}_{m}^{(i+2)}$, which have not been computed yet. However, the last rows and columns of $\bar{\mathbf{\Psi}}_{m}^{(i+2)}$ and $\bar{\mathbf{\Gamma}}_{m}^{(i+2)}$ do not influence the calculation since these rows and columns are zeroed out according to the forms of $\bar{\mathbf{\Phi}}_{m}^{(i+2)}$ and $\bar{\mathbf{\Delta}}_{m}^{(i+2)}$. In other words, we can just use the upper-left submatrices of $\bar{\mathbf{\Psi}}_{m}^{(i+2)}$ and $\bar{\mathbf{\Gamma}}_{m}^{(i+2)}$, i.e., $\bar{\mathbf{\Psi}}_{m}^{(i+1)}$ and $\bar{\mathbf{\Gamma}}_{m}^{(i+1)}$, to obtain $\mathbf{\Theta}_{m,j}^{(i+2)}$ and $\bar{\mathbf{M}}^{(i+2)}_{m,\gamma}$. In addition, all the expectations in these matrices are calculated via empirical averaging such that $\mathbf{\Psi}_{m}^{(i)}=\bar{\mathbf{\Psi}}_{m}^{(i)}$, $\mathbf{\Phi}_{m}^{(i)}=\bar{\mathbf{\Phi}}_{m}^{(i)}$, and $\mathbb{E}[\tilde{H}^{(i)}_{m}\tilde{H}^{(j)}_{m}]$, $\mathbb{E}[H_{m}\tilde{H}^{(i)}_{m}]$, $\mathbb{E}[H_{m}^2]$, and $\mathbb{E}[S^{(i)}_{m}S^{(j)}_{m}]$ can be obtained as $\frac{(\tilde{\mathbf{h}}_{m}^{(i)})^{\mathrm{H}}\tilde{\mathbf{h}}_{m}^{(j)}}{(T+1)N}$, $\frac{||\tilde{\mathbf{h}}^{(i)}_{m}||_2^2}{(T+1)N}$, $\frac{\lambda}{N}\sum_{n=1}^{N}\beta_{n}$, and $\frac{(\mathbf{s}_{m}^{(i)})^{\mathrm{H}}\mathbf{s}_{m}^{(j)}}{L}$, respectively. Besides, $\{\bar{\kappa}_{2k}\}_{k\geq1}$ directly uses the values of $\{\kappa_{2k}\}_{k\geq1}$.

Similar to the AMP algorithm, we adopt the MMSE denoisers, while intermediate estimates of all the preceding iterates instead of only the most recent one are utilized. Specifically, the denoisers are expressed as follows:
\begin{align}
    f^{(i)}(\mathbf{h}_{m}^{(1)},\cdots,\mathbf{h}_{m}^{(i)})=\mathbb{E}[\mathbf{h}_{m}|\mathbf{h}_{m}^{(1)},\cdots,\mathbf{h}_{m}^{(i)}],
\label{fpampffunctiondefine}
\end{align}
\begin{align}
\begin{aligned}
    g^{(i+1)}(\mathbf{r}_{m}^{(1)},\cdots,\mathbf{r}_{m}^{(i)},\mathbf{y}_{m})
    &=\mathbb{E}[\mathbf{G}_{m}|\mathbf{r}_{m}^{(1)},\cdots,\mathbf{r}_{m}^{(i)},\mathbf{y}_{m}]\\
    &-\mathbb{E}[\mathbf{G}_{m}|\mathbf{r}_{m}^{(1)},\cdots,\mathbf{r}_{m}^{(i)}].
\label{fpampgfunctiondefine}
\end{aligned}
\end{align}

Regarding the denoiser in (\ref{fpampffunctiondefine}), we need to compute the posterior distribution of $\mathbf{h}_{m}$. Since the likelihood function can be obtained in (\ref{SE1}) and the prior information is available in (\ref{prior}), the posterior distribution is calculated by using the Bayes' theorem as follows:
\begin{align}
\begin{aligned}
&\quad p\left(\mathbf{h}_{n,m}|\mathbf{h}_{n,m}^{(1)},\cdots,\mathbf{h}_{n,m}^{(i)}\right)\\
&=\left(1-\sum_{t=0}^{T}\pi_{n,t,m}^{(i)}\right)\prod_{t=0}^{T}\delta_{0}\left(h_{n, t,m}\right)\\
&+\sum_{t=0}^{T}\left(\pi_{n,t,m}^{(i)}\mathcal{C N}\left(h_{n,t,m}; C_{n,t,m}^{(i)}, D_{n,t,m}^{(i)}\right)\prod_{t^{\prime}\neq t}\delta_0(h_{n,t^{\prime},m})\right),
\end{aligned}
\label{fpampposterior}
\end{align}
where 
\begin{align}
    C_{n,t,m}^{(i)}=\frac{E_{n,t,m}^{(i)}}{F_{n,m}^{(i)}},
\label{fpampc}
\end{align}
\begin{align}
    D_{n,t,m}^{(i)}=\frac{1}{F_{n,m}^{(i)}},
\label{fpampd}
\end{align}
\begin{align}
    E_{n,t,m}^{(i)}=(\bar{\boldsymbol{\mu}}_{m}^{(i)})^{\mathrm{H}}(\bar{\mathbf{\Omega}}_{m}^{(i)})^{-1}\left[h_{n,t,m}^{(1)},\cdots,h_{n,t,m}^{(i)}\right]^{\mathrm{T}},
\label{fpampe}
\end{align}
\begin{align}
    F_{n,m}^{(i)}=1/\beta_{n}+(\bar{\boldsymbol{\mu}}_{m}^{(i)})^{\mathrm{H}}(\bar{\mathbf{\Omega}}_{m}^{(i)})^{-1}\bar{\boldsymbol{\mu}}_{m}^{(i)},
\label{fpampf}
\end{align}
\begin{align}
\pi_{n,t,m}^{(i)}=\frac{e^{\frac{\left(E_{n,t,m}^{(i)}\right)^2}{F_{n,m}^{(i)}}}}{\beta_n F_{n,m}^{(i)}\frac{1-\sum_{t^{\prime}=0}^{T}\lambda_{n,t^{\prime},m}}{\lambda_{n,t,m}}+\sum_{t^{\prime}=0}^{T}\frac{\lambda_{n,t^{\prime},m}}{\lambda_{n,t,m}}e^{\frac{\left(E_{n,t^{\prime},m}^{(i)}\right)^2}{F_{n,m}^{(i)}}}}.
\label{fpamppostpara}
\end{align}
Therefore, each element in $\tilde{\mathbf{h}}_{m}^{(i)}$ of (\ref{fpamp2}) in $i$-th iteration, i.e, $\tilde{h}_{n,t,m}^{(i)}$, can be obtained as follows:
\begin{align}
\tilde{h}_{n,t,m}^{(i)} = [f^{(i)}(\mathbf{h}_{n,m}^{(1)},\cdots,\mathbf{h}_{n,m}^{(i)})]_{t}=\pi_{n,t,m}^{(i)}C_{n,t,m}^{(i)}.
\label{ffunction}
\end{align}
After some manipulations, $\partial_{h_{n,t,m}^{(k)}} [f^{(i)}(\mathbf{h}_{n,m}^{(1)},\cdots,\mathbf{h}_{n,m}^{(i)})]_t$ is derived as follows:
\begin{align}
\begin{aligned}
    &\partial_{h_{n,t,m}^{(k)}} [f^{(i)}(\mathbf{h}_{n,m}^{(1)},\cdots,\mathbf{h}_{n,m}^{(i)})]_t = \\
    &\pi_{n,t,m}^{(i)}\left[\frac{2E_{n,t,m}^{(i)}}{F_{n,m}^{(i)}}C_{n,t,m}^{(i)}\left(1-\pi_{n,t,m}^{(i)}\right)+1\right](\bar{\boldsymbol{\mu}}_{m}^{(i)})^{\mathrm{H}}(\bar{\mathbf{\Omega}}_{m}^{(i)})^{-1}\mathbf{e}_{k},
\end{aligned}
\label{partialf}
\end{align}
where $\mathbf{e}_{k} \in \mathbb{R}^{i \times 1}$ represents the $k$-th column of $\mathbf{I}_{i}$.

Regarding the denoiser in (\ref{fpampgfunctiondefine}), since (\ref{SMVmodel}) is a linear model, we set $g^{(1)}(\mathbf{y}_{m})=\mathbf{y}_{m}$ and consequently $\partial_{G_{l,m}} g^{(1)}(y_{l,m}) = 1$ for $i=0$. For $i>0$, based on the state evolution, we have $[G_m,R_{m}^{(1)},\cdots,R_{m}^{(i)}]^{\mathrm{T}} \sim \mathcal{CN}(\mathbf{0},\bar{\mathbf{\Sigma}}_{m}^{(i+1)})$, and the second conditional expectation in (\ref{fpampgfunctiondefine}) is given as follows:
\begin{align}
\begin{aligned}
    &\mathbb{E}[G_{l,m}|r_{l,m}^{(1)},\cdots,r_{l,m}^{(i)}] \\
    &= [\bar{\mathbf{\Sigma}}_{m}^{(i+1)}]_{1,2:i+1}([\bar{\mathbf{\Sigma}}_{m}^{(i+1)}]_{2:i+1,2:i+1})^{-1}\left[r_{l,m}^{(1)},\cdots,r_{l,m}^{(i)}\right]^{\mathrm{T}}.
\label{gsecond}
\end{aligned}
\end{align}
By further incorporating the received pilot signal, we have $[G_m,R_{m}^{(1)},\cdots,R_{m}^{(i)},Y_m]^{\mathrm{T}} \sim \mathcal{CN}(\mathbf{0},\bar{\mathbf{S}}_{m}^{(i+2)})$, where
\begin{align}
    \bar{\mathbf{S}}_{m}^{(i+2)}= \left[\begin{array}{cc}
\bar{\mathbf{\Sigma}}_{m}^{(i+1)} & \left[\bar{\mathbf{\Sigma}}_{m}^{(i+1)}\right]_{1: i+1,1} \\
\left[\bar{\mathbf{\Sigma}}_{m}^{(i+1)}\right]_{1,1: i+1} & \kappa_2\mathbb{E}[H_{m}^2]+\frac{\sigma^2}{\rho L}
\end{array}\right].
\label{Sbar}
\end{align}
Thus, the first conditional expectation in (\ref{fpampgfunctiondefine}) is expressed as follows:
\begin{align}
\begin{aligned}
    &\mathbb{E}[G_{l,m}|r_{l,m}^{(1)},\cdots,r_{l,m}^{(i)}, y_{l,m}] \\
    &= [\bar{\mathbf{S}}_{m}^{(i+2)}]_{1,2:i+2}([\bar{\mathbf{S}}_{m}^{(i+2)}]_{2:i+2,2:i+2})^{-1}\left[r_{l,m}^{(1)},\cdots,r_{l,m}^{(i)},y_{l,m}\right]^{\mathrm{T}}.
\label{gfirst}
\end{aligned}
\end{align}
By combining the results in (\ref{gsecond}) and (\ref{gfirst}), we obtain
\begin{align}
\begin{aligned}
    &g^{(i+1)}(r_{l,m}^{(1)},\cdots,r_{l,m}^{(i)},y_{l,m})\\
    &=[\bar{\mathbf{S}}_{m}^{(i+2)}]_{1,2:i+2}([\bar{\mathbf{S}}_{m}^{(i+2)}]_{2:i+2,2:i+2})^{-1}\left[r_{l,m}^{(1)},\cdots,r_{l,m}^{(i)},y_{l,m}\right]^{\mathrm{T}}\\
    &-[\bar{\mathbf{\Sigma}}_{m}^{(i+1)}]_{1,2:i+1}([\bar{\mathbf{\Sigma}}_{m}^{(i+1)}]_{2:i+1,2:i+1})^{-1}\left[r_{l,m}^{(1)},\cdots,r_{l,m}^{(i)}\right]^{\mathrm{T}}.
\end{aligned}
\label{gfunction}
\end{align}
Accordingly, the partial derivatives of $g^{(i+1)}$ can be calculated as follows:
\begin{align}
\begin{aligned}
    &\partial_{G_{l,m}} g^{(i+1)}(r_{l,m}^{(1)},\cdots,r_{l,m}^{(i)},y_{l,m})\\
    &= [\bar{\mathbf{S}}_{m}^{(i+2)}]_{1,2:i+2}([\bar{\mathbf{S}}_{m}^{(i+2)}]_{2:i+2,2:i+2})^{-1}\mathbf{e}_{i+1},
\label{partialgg}
\end{aligned}
\end{align}
\begin{align}
\begin{aligned}
    &\partial_{r_{l,m}^{(k)}} g^{(i+1)}(r_{l,m}^{(1)},\cdots,r_{l,m}^{(i)},y_{l,m})\\
    &= [\bar{\mathbf{S}}_{m}^{(i+2)}]_{1,2:i+2}([\bar{\mathbf{S}}_{m}^{(i+2)}]_{2:i+2,2:i+2})^{-1}\left[\mathbf{e}_{k}^{\mathrm{T}},0\right]^{\mathrm{T}}\\
    &-[\bar{\mathbf{\Sigma}}_{m}^{(i+1)}]_{1,2:i+1}([\bar{\mathbf{\Sigma}}_{m}^{(i+1)}]_{2:i+1,2:i+1})^{-1}\mathbf{e}_{k}.
\end{aligned}
\label{partialgr}
\end{align}

Similar to the OAMP-based algorithm, the common sparsity pattern among all antennas is leveraged in the FPAMP-based algorithm. When the FPAMP iterations terminate, other steps follow those of the proposed OAMP-based algorithm in Section \ref{extendtoMMV}. The proposed FPAMP-based algorithm for joint activity-delay detection and channel estimation are summarized in Algorithm \ref{algo2}.

\begin{algorithm}[t]
\caption{The Proposed FPAMP-based Algorithm \label{algo2}}
{\bf Input:}
The normalized received pilot signal $\mathbf{Y}$, pilot sequences $\mathbf{P}$, maximum synchronization
delay $T$, maximum number of iterations $Q$, and accuracy tolerance $\epsilon$.\\
{\bf Output:}
The estimated effective channel matrix $\hat{\mathbf{H}}$, set of active users $\hat{\mathcal{K}}$, and synchronization delays $\{t_{n}\}$'s, $n \in \hat{\mathcal{K}}$.
\begin{algorithmic}[1]
\STATE Calculate the rectangular free cumulants $\{\kappa_{2k}\}_{k \geq 1}$ with (\ref{freecumulant}).\\
{\bf Initialize:}
$i \leftarrow 0$, $\mathbf{s}_{m}^{(1)}=\mathbf{y}_{m}$, $\forall m \in \mathcal{M}$, $\mathbf{h}_{m}^{(1)}=\mathbf{P}^{\mathrm{H}}\mathbf{s}_{m}^{(1)}$, $\forall m \in \mathcal{M}$, $\omega_{n,t}^{(1)}=\frac{\lambda}{T+1}$, $\forall n \in \mathcal{N}$, $\forall t \in \{0,\cdots,T\}$, $\bar{\boldsymbol{\mu}}_{m}^{(1)}=\frac{L}{(T+1)N}\kappa_2$, $\forall m \in \mathcal{M}$, $\bar{\mathbf{\Sigma}}_{m}^{(1)}=\kappa_2\frac{\lambda}{N}\sum_{n=1}^{N}\beta_{n}$, $\forall m \in \mathcal{M}$, and $\bar{\mathbf{\Omega}}_{m}^{(1)}=\frac{1}{(T+1)N}\kappa_2\mathbf{y}_{m}^{\mathrm{H}}\mathbf{y}_{m}+\frac{L}{(T+1)N}\kappa_4\frac{\lambda}{N}\sum_{n=1}^{N}\beta_{n}$, $\forall m \in \mathcal{M}$.
\WHILE{{$i < Q$} \text{and} {$\frac{\sum_{m}||\tilde{\mathbf{h}}_{m}^{(i)}-\tilde{\mathbf{h}}_{m}^{(i-1)}||_{2}^{2}}{\sum_{m}||\tilde{\mathbf{h}}_{m}^{(i-1)}||_{2}^{2}} > \epsilon$}}
\STATE $i \leftarrow i+1$
\STATE Calculate the auxiliary matrices $\bar{\mathbf{\Psi}}_{m}^{(i+1)}$, $\bar{\mathbf{\Phi}}_{m}^{(i+1)}$, $\bar{\mathbf{\Gamma}}_{m}^{(i+1)}$
\Statex \quad and $\bar{\mathbf{\Delta}}_{m}^{(i+1)}$, $\forall m$ according to (\ref{psibar}), (\ref{phibar}), (\ref{Gammabar}) and (\ref{Deltabar}).
\STATE Calculate the matrices $\mathbf{M}_{m,\gamma}^{(i+1)}$ and the Onsager coeffi-
\Statex \quad cients $\{\gamma_{m,j}^{(i)}\}_{j=1}^{i-1}$, $\forall m$ according to (\ref{Mgamma}) and (\ref{gammaupdate}), where 
\Statex \quad $\mathbf{\Psi}_{m}^{(i+1)}=\bar{\mathbf{\Psi}}_{m}^{(i+1)}$ and $\mathbf{\Phi}_{m}^{(i+1)}=\bar{\mathbf{\Phi}}_{m}^{(i+1)}$, $\forall m$.
\STATE Perform (\ref{fpamp1}), $\forall m$.
\STATE Update the prior sparsity ratio $\omega_{n,t}^{(i)}$, $\forall n, t$ according to
\Statex \quad (\ref{sparsityratio}).
\STATE Perform (\ref{fpamp2}), $\forall m$ according to (\ref{fpampc})-(\ref{ffunction}), where $\lambda_{n,t,m}$
\Statex \quad is replaced by $\omega_{n,t}^{(i)}$.
\STATE Calculate the matrices $\mathbf{M}_{m,\alpha}^{(i+1)}$ and the Onsager coeffi-
\Statex \quad cients $\{\alpha_{m,j}^{(i)}\}_{j=1}^{i}$, $\forall m$ according to (\ref{Malpha}) and (\ref{alphaupdate}), where 
\Statex \quad $\mathbf{\Psi}_{m}^{(i+1)}=\bar{\mathbf{\Psi}}_{m}^{(i+1)}$ and $\mathbf{\Phi}_{m}^{(i+1)}=\bar{\mathbf{\Phi}}_{m}^{(i+1)}$, $\forall m$.
\STATE Perform (\ref{fpamp3}), $\forall m$.
\STATE Calculate the covariance matrix $\bar{\mathbf{\Sigma}}_{m}^{(i+1)}$, $\forall m$ for the
\Statex \quad denoiser $g^{(i+1)}$ according to (\ref{Sigmaupdate}) and (\ref{Sigmaupdate1}).
\STATE Perform (\ref{fpamp4}), $\forall m$ according to (\ref{Sbar}) and (\ref{gfunction}).
\STATE Calculate the mean vector $\bar{\boldsymbol{\mu}}_{m}^{(i+1)}$, $\forall m$ according to (\ref{muupdate}) 
\Statex \quad and (\ref{Mbarupdate}).
\STATE Calculate the covariance matrix $\bar{\mathbf{\Omega}}_{m}^{(i+1)}$, $\forall m$ for the 
\Statex \quad denoiser $f^{(i+1)}$ according to (\ref{Omegaupdate1}), (\ref{Omegaupdate2}) and (\ref{Omegaupdate3}).
\STATE Calculate the partial derivatives $\{\partial_{h_{n,t,m}^{(k)}} [f^{(i)}(\mathbf{h}_{n,m}^{(1)},\cdots,$ 
\Statex \quad $\mathbf{h}_{n,m}^{(i)})]_t\}$, $\{\partial_{G_{l,m}} g^{(i+1)}(r_{l,m}^{(1)},\cdots,r_{l,m}^{(i)},y_{l,m})\}$ and
\Statex \quad $\{\partial_{r_{l,m}^{(k)}} g^{(i+1)}(r_{l,m}^{(1)},\cdots,r_{l,m}^{(i)},y_{l,m})\}$, $\forall m$ according to
\Statex \quad (\ref{partialf}), (\ref{partialgg}) and (\ref{partialgr}), where $\lambda_{n,t,m}$ is replaced by $\omega_{n,t}^{(i)}$.
\ENDWHILE
\STATE Obtain $\hat{\mathbf{H}}=[\hat{\mathbf{h}}_{1},\cdots,\hat{\mathbf{h}}_{M}]$ with $\hat{\mathbf{h}}_{m}=\tilde{\mathbf{h}}_{m}^{(I)}$.
\STATE Determine the set of active users as $\hat{\mathcal{K}}\triangleq \{n\in \mathcal{N}|\sum_{t=0}^{T} \omega_{n,t}^{(I+1)}\geq \theta\}$.
\STATE Determine the synchronization delays as $t_{n}=\mathop{\arg\max}\limits_{t \in \{0,\cdots,T\}} \omega_{n,t}^{(I+1)}$, $\forall n$.
\end{algorithmic}
\end{algorithm}

\subsection{Computational Complexity Analysis \label{sectioncomplexity}}
The computational complexity of the proposed OAMP-based and FPAMP-based algorithms is analyzed and compared with that of the AMP-based algorithm \cite{wzhu2021}. Since the deep unrolling method in \cite{wzhu2021} is not able to reduce the complexity of the AMP algorithm significantly, it is not considered in the analysis. As shown in Table \ref{complexitytable}, we choose the number of complex-valued multiplications as the metric of interest, and the complexity of one real-valued multiplication is assumed to be one quarter of the complex-valued counterpart's. For OAMP-based algorithm that uses LMMSE, the computational complexity is highest due to the matrix inverse operation (\ref{decor}) in each iteration. Our proposed FPAMP-based algorithm avoids the matrix inversion and thus enjoys a low complexity. For iteration with index $i \ll N(T+1)$, the FPAMP-based and AMP-based algorithms have comparable complexity. Our experimental results show, the maximum iteration number for convergence of FPAMP-based algorithm typically does not exceed $20$, which is far smaller than the value of $(T+1)N$ in grant-free massive RA systems.

We use the simulation setting with $K = 50$, $M = 16$, $N=400$, $\bar{L} = 50$, and $T=4$, as will be detailed in Section \ref{sectionv}, to give an intuitive picture on the computational complexity of these algorithms. The numbers of complex-valued multiplications of the AMP and OAMP-based algorithms in one iteration are given by $2.2\times 10^6$ and $9.4 \times 10^7$, respectively. In addition, suppose $i$ is swept from $1$ to $10$, the number of complex-valued multiplications of the FPAMP algorithm in the $i$-th iteration increases from $1.8\times 10^6$ to $4.2\times 10^6$. These results confirm that the FPAMP-based algorithm has a much lower complexity compared with the OAMP-based algorithm.
\begin{table}[t]
\caption{Computational Complexity Analysis}
\centering
\scalebox{1}{
    \begin{tabular}{c|c}
    \hline
    Algorithm & Number of complex multiplications in each iteration\\
    \hline
    AMP-based & $\frac{1}{4}M\left(5(T+1)NL+9(T+1)N+3L\right) $\\
    \hline
    OAMP-based & $\frac{1}{4}M\left(4(T+1)NL^2\!+\!L^3+2(T+1)NL+4(T+1)N\right)$\\
    \hline
    FPAMP-based & \makecell{$\frac{1}{4}M\Big(4(T+1)NL+3(T+1)N(i+1)$\\ $+L(3i+2)+12i^3+9i^2+10i$\\$+(4i^3+4i^2)\times(\sum_{k=1}^{2i}\log_{2}k+\log_{2}(2i+1-k))\Big)$}\\
    \hline
    \end{tabular}}
\label{complexitytable}
\end{table}

\section{Simulation Results \label{sectionv}}
We consider an uplink cellular network where 400 users are uniformly distributed within a circular ring centered at a multi-antenna BS in our simulations. The path loss of user $n$ is modeled as $\beta_{n}=-128.1-36.7\log_{10}(d_{n})$ (dB) with $d_{n} \in [0.05, 1]$ km. The number of BS antennas is $M=16$ and the pilot sequence length is $\bar{L}=50$. Besides, the transmit power of each user is set to be $23$ dBm, and the noise power spectrum density is $-169$ dBm/Hz over $1$ MHz bandwidth. In addition, the maxmium number of iterations is $Q=50$, and the accuracy tolerance is $\epsilon=10^{-5}$. The simulation results are averaged over $10^5$ independent channel and active user set realizations. By default, the number of active users $K$, the transmit power, and the maximum delay $T$ are set as $50$, $23$ dBm, and $4$ symbols, respectively. For comparisons, the following three baselines are also simulated:
\begin{itemize}
    \item \textbf{Group LASSO-based algorithm \cite{lliu2021}:} This method formulates the joint activity-delay detection and channel estimation problem as a group LASSO problem, which is solved by a block coordinate descent algorithm. 
\end{itemize}
\begin{itemize}
    \item \textbf{AMP-based algorithm \cite{wzhu2021}:} This method uses the AMP algorithm with an MMSE denoiser to perform joint activity-delay detection and channel estimation. However, we do not consider the deep learning implementation in \cite{wzhu2021} due to the marginal gain of detection and estimation accuracy.
\end{itemize}
\begin{itemize}
    \item \textbf{Covariance-based algorithm \cite{zwang2022}:} This method employs a block coordinate descent algorithm based on the covariance matrix of the received pilot signal to first perform activity and delay detection, after which, the channel coefficients are estimated via an MMSE algorithm.
\end{itemize}

\begin{figure}[t]
\centering
\includegraphics[width=3.4in]{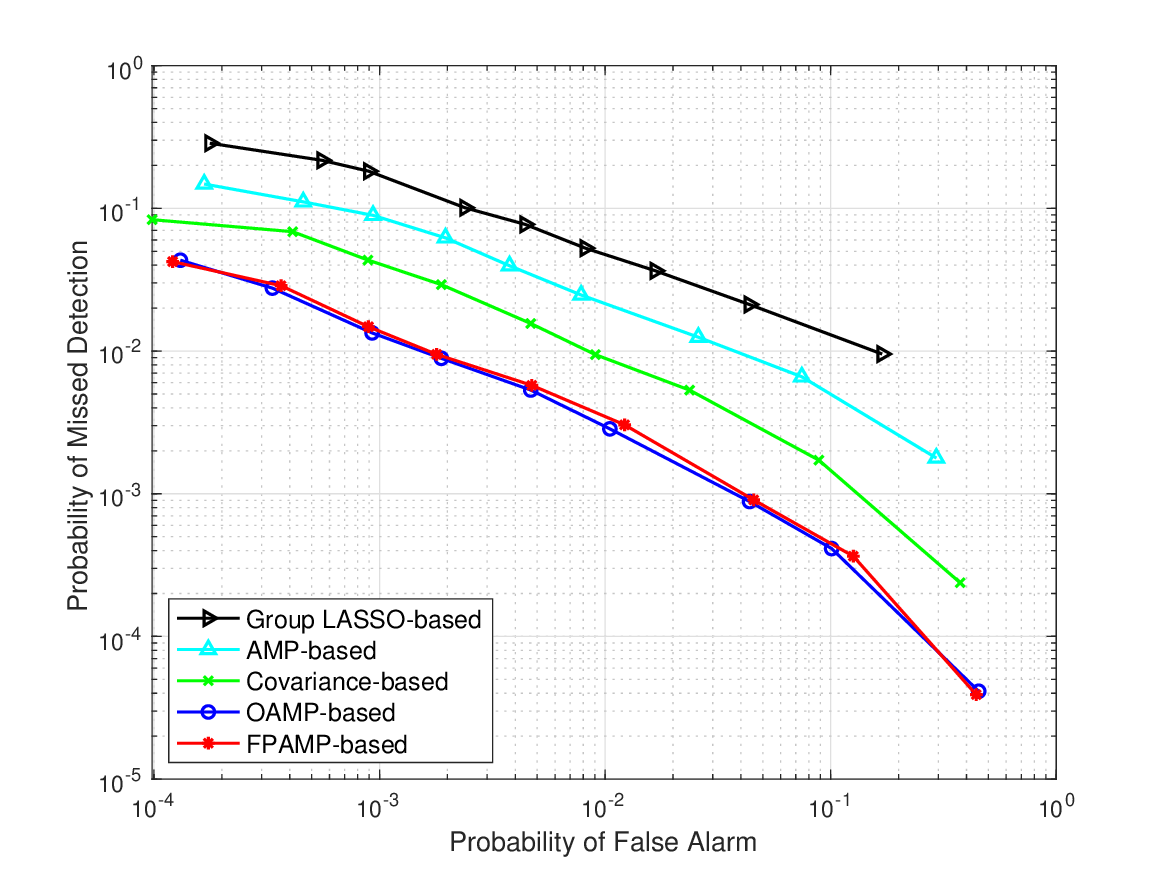}
\caption{Probability of missed detection versus probability of false alarm.}
\label{AUD}
\end{figure}
We first evaluate the user activity detection performance of different algorithms. The probability of missed detection versus the probability of false alarm is shown in Fig. \ref{AUD} by tuning the threshold $\theta$ for determining the set of active users. In particular, the probability of missed detection is defined as $K_m \slash K$, where $K_m$ denotes the number of active users that are detected as inactive, and the probability of false alarm is defined as $N_{f} \slash \left(N-K\right)$ with $N_f$ denoting the number of inactive users that are detected as active. It is first observed from Fig. \ref{AUD} that the two proposed algorithms achieve the lowest missed detection probabilities for a given false alarm probability. Besides, although the covariance-based method outperforms the group LASSO-based and AMP-based methods, it is outperformed by the two proposed algorithms by a large margin. Such performance improvement, on one hand, is attributed to the capabilities of the proposed OAMP-based and FPAMP-based algorithms in handling pilot matrices with non-i.i.d. entries. On the other hand, exploration of the common sparsity pattern among multiple BS antennas also contributes significantly. In addition, the proposed FPAMP-based algorithm has negligible performance degradation compared with the OAMP-based algorithm despite with reduced computational complexity.

Next, the delay detection error probability is shown in Fig. \ref{DelayvK}, which is defined as $K_{d} \slash K$ with $K_d$ denoting the number of active users with wrongly detected synchronization delay. The delay detection error probabilities are obtained with a fixed false alarm probability of $10^{-1}$, which is achieved by tuning the threshold value $\theta$. It is seen that an increased number of active users leads to degraded delay detection performance due to the limited radio resources for pilot transmissions. In accordance with in Fig. \ref{AUD}, the two proposed algorithms achieve the best performance, which results from their supreme activity detection performance since the delay detection procedure is based on the channel estimation and activity detection. Since the successful transmission critically depends on correct synchronization delay detection, and the missed detection probability is lower bounded by the delay detection error probability with fixed false alarm probability. Hence, we adopt the delay detection error probability as an indicator on the capability of supporting active users in asynchronous massive RA systems. Specifically, assuming the probability of missed detection requirement is $2\times 10^{-3}$, the two proposed algorithms are able to support 60 active users while the AMP-based algorithm can only support 34, which is a remarkable 76\% increase.
\begin{figure}[t]
\centering
\includegraphics[width=3.4in]{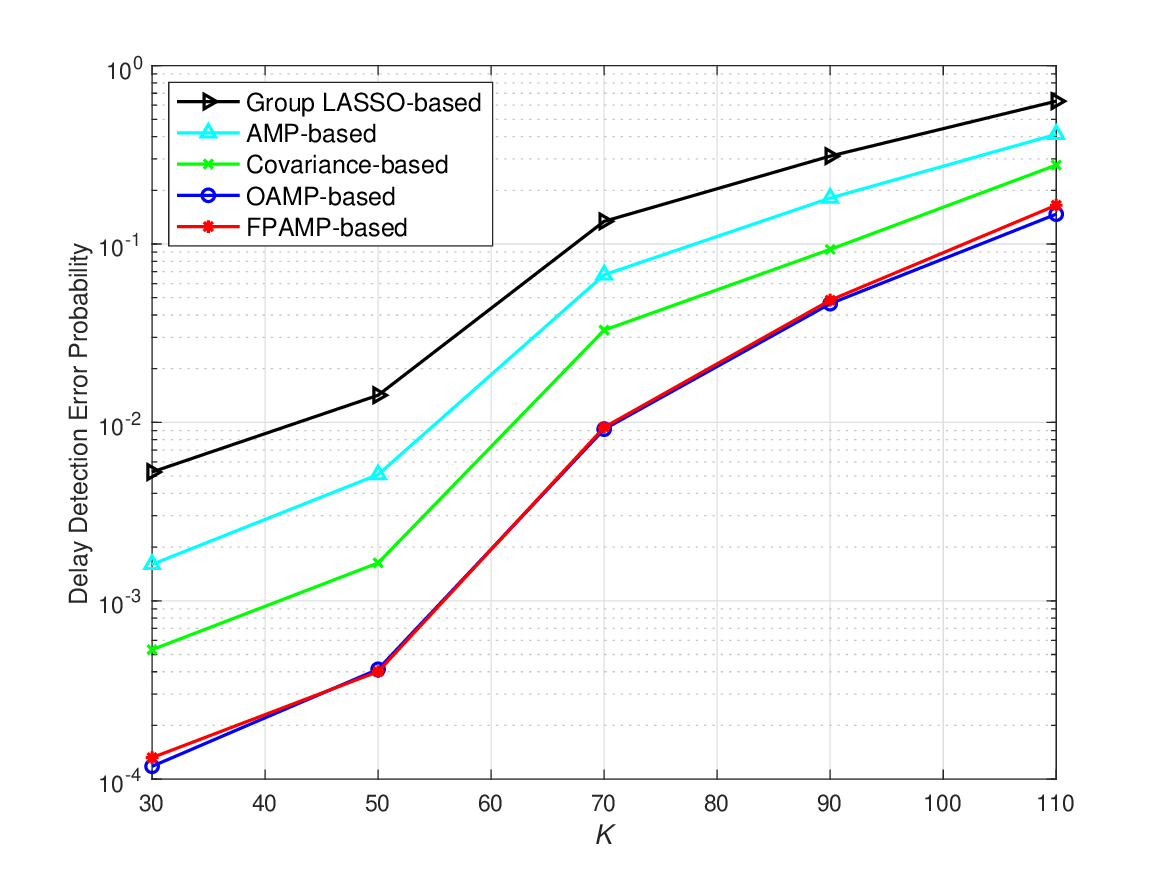}
\caption{Delay detection error probability versus the number of active users $K$.}
\label{DelayvK}
\end{figure}

We investigate the relationship between the normalized mean square error (NMSE) of channel estimation and transmit power in Fig. \ref{NMSEvP}. Similar to Figs. \ref{AUD} and \ref{DelayvK}, the two proposed algorithms achieve significant NMSE reduction compared with the baselines, especially with large transmit power, and the FPAMP-based algorithm maintains similar performance as the OAMP-based algorithm. These observations again validate the benefits of the LMMSE estimator in the OAMP-based algorithm in eliminating multi-user interference, and using rectangular free cumulants of the non-i.i.d Gaussian pilot matrix. In particular, the latter enables the compatibility of the low-complexity AMP framework with general pilot matrices.
\begin{figure}[t]
\centering
\includegraphics[width=3.4in]{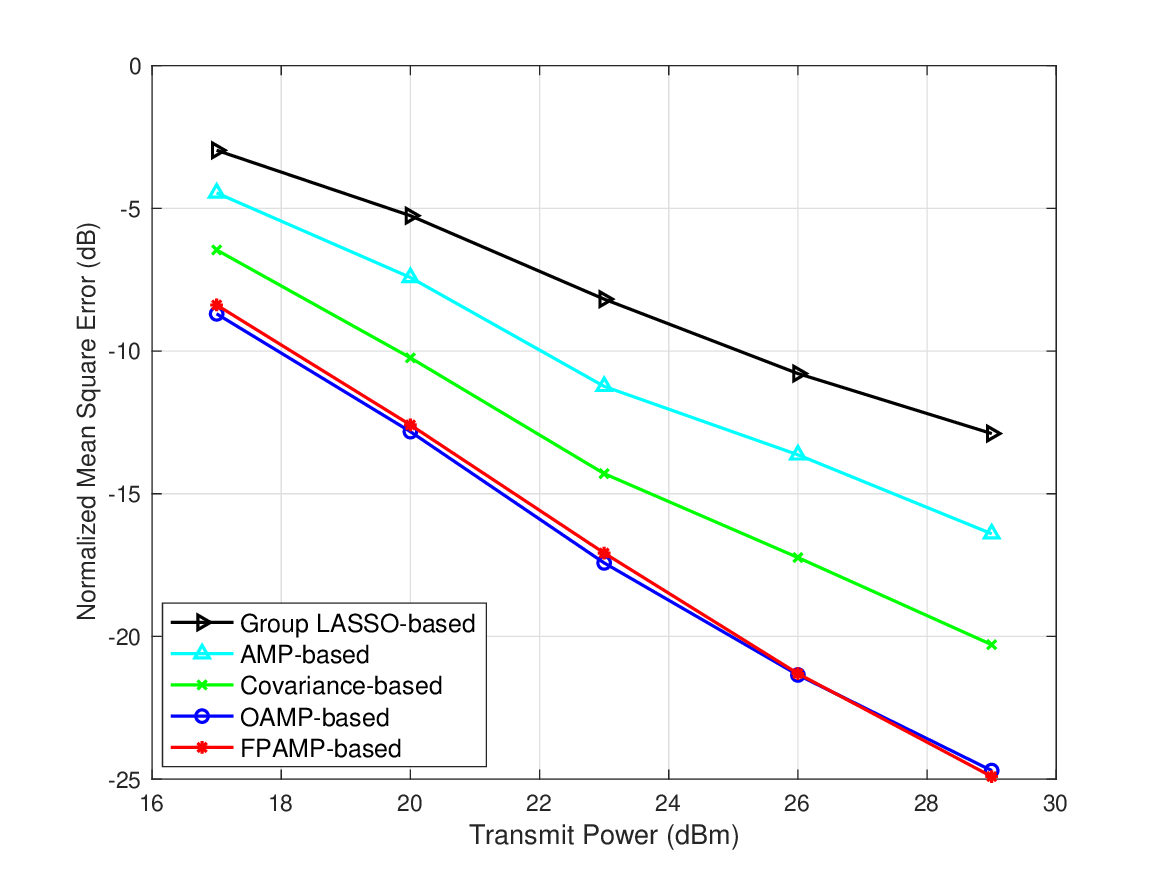}
\caption{NMSE of channel estimation versus transmit power.}
\label{NMSEvP}
\end{figure}

The user activity detection performance with different maximum synchronization delay, i.e., $T=4$ and $T=9$, is further examined in Fig. \ref{AUDwithT}. We find that the performance of all methods deteriorates with a larger synchronization delay since more zero elements exist in the expanded pilot matrix that destroys the Gaussianity to a greater extent. Furthermore, the performance of the two proposed algorithms is less affected compared with the AMP-based algorithm, indicating that they are more robust to the non-i.i.d. Gaussian pilot matrices.
\begin{figure}[t]
\centering
\includegraphics[width=3.4in]{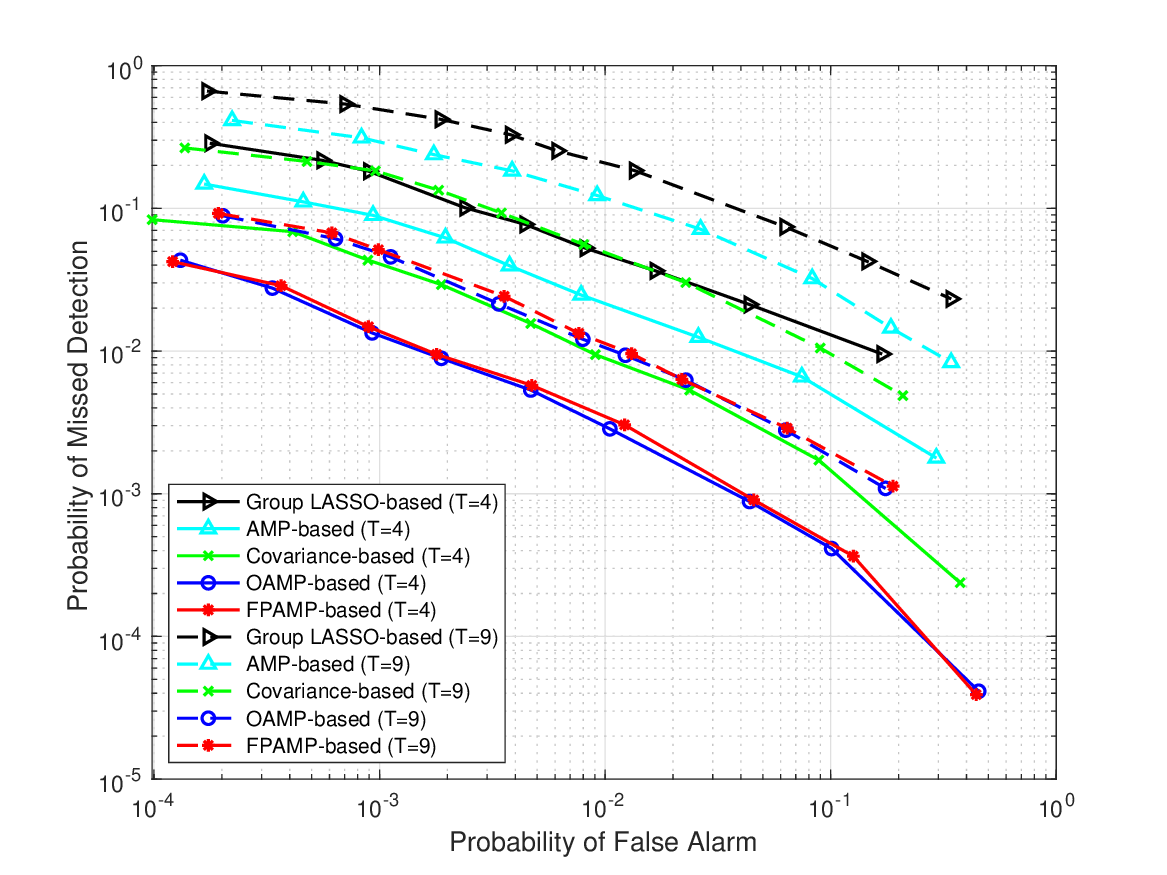}
\caption{Probability of missed detection versus probability of false alarm with different maximum synchronization delays.}
\label{AUDwithT}
\end{figure}

The computational complexity of different algorithms in terms of average execution time on an Apple Macbook Pro laptop (Model: 14 inch, M1 Pro Chip with 16GB RAM) is summarized in Table \ref{time}. From the table, it is clear that the group LASSO-based method has the lowest complexity. However, its performance is far inferior to that of other methods. Besides, while the OAMP-based algorithm achieves the best detection and estimation accuracy, it suffers from heavy computational overhead that originates from the LMMSE estimator. Although the covariance-based method has similar computational complexity as the OAMP-based method, there exists a significant performance gap as previously shown in Figs. \ref{AUD} $\sim$ \ref{NMSEvP}. In addition, the AMP-based and FPAMP-based algorithms have comparable average execution time, and the FPAMP-based algorithm secures approximately 41\% complexity reduction compared with the OAMP-based algorithm. This demonstrates that the FPAMP-based algorithm can fully exploit the property of pilot matrices with their the rectangular free cumulants to achieve low complexity and satisfactory detection/estimation accuracy.
\begin{table}[t]
\caption{Average execution time of different algorithms}
\centering
\begin{tabular}{c|c}
    \hline
     Algorithm & Average execution time (s)  \\
    \hline
    Group LASSO-based & 2.78  \\
    \hline
    AMP-based & 3.27 \\
    \hline
    Covariance-based & 5.11 \\
    \hline
    OAMP-based & 5.92\\
    \hline
    FPAMP-based & 3.51\\
    \hline
    \end{tabular}
\label{time}
\end{table}

\section{Conclusions \label{sectionvi}}
In this paper, we investigated the joint activity detection, synchronization delay detection, and channel estimation in asynchronous grant-free massive random access (RA) systems. Considering that entries in the pilot matrix are not independent and identically Gaussian distributed, we first proposed a novel algorithm based on orthogonal approximate message passing (OAMP), which also fully utilizes the common sparsity among the received pilot signals of multiple base station antennas. To accelerate the computation, a free probability approximate message passing (FPAMP)-based algorithm was further developed, which explores the rectangular free cumulants of the pilot matrix to avoid matrix inversion. Simulation results showed the effectiveness of the proposed algorithms, and the potential of the FPAMP-based algorithm for fast joint activity-delay detection and channel estimation in grant-free massive RA systems. 

Our study demonstrates the benefits of exploiting high-order statistical information from the pilot matrix when designing the detection and estimation algorithm for grant-free massive RA. For future investigations, it would be interesting to develop more advanced algorithms for asynchronous massive RA systems, e.g., extending the proposed algorithms by fusing deep unrolling and FPAMP, to further enhance the performance and reduce the complexity.

\appendices
\section{Proof of Theorem \ref{theoremconverge} \label{appendixa}}
In order to prove the theorem, we first present the general recursion as follows:
\begin{align}
\mathbf{z}_{m}^{(i)} = \mathbf{P}^{\mathrm{H}}\mathbf{u}_{m}^{(i)}-\sum_{j=1}^{i-1}b_{m,j}^{(i)}\mathbf{v}_{m}^{(j)},
\label{generalrecursion1}
\end{align}
\begin{align}
\mathbf{v}_{m}^{(i)} = \tilde{f}^{(i)}(\mathbf{z}_{m}^{(1)},\cdots,\mathbf{z}_{m}^{(i)},\mathbf{E}_{m}),
\label{generalrecursion2}
\end{align}
\begin{align}
\mathbf{o}_{m}^{(i)} = \mathbf{P}\mathbf{v}_{m}^{(i)}-\sum_{j=1}^{i}a_{m,j}^{(i)}\mathbf{u}_{m}^{(j)},
\label{generalrecursion3}
\end{align}
\begin{align}
\mathbf{u}_{m}^{(i+1)} = \tilde{g}^{(i+1)}(\mathbf{o}_{m}^{(1)},\cdots,\mathbf{o}_{m}^{(i)},\mathbf{F}_{m}),
\label{generalrecursion4}
\end{align}
    
\noindent where $\mathbf{u}_{m}^{(1)}=\mathbf{0}$, $\mathbf{z}_{m}^{(1)}=\mathbf{P}^{\mathrm{H}}\mathbf{u}_{m}^{(1)}$, and $\mathbf{E}_{m}=\mathbf{h}_{m}$ and $\mathbf{F}_{m}=\mathbf{n}_{m}$ are side information independent of $\mathbf{P}$ with finite second-order moments. Functions $\tilde{f}^{(i)}: \mathbb{C}^{i+1}\rightarrow \mathbb{C}$ and $\tilde{g}^{(i+1)}: \mathbb{C}^{i+1}\rightarrow \mathbb{C}$ are pseudo-Lipschitz. In particular, $\tilde{f}^{(1)}(\mathbf{z}_{m}^{(1)},\mathbf{h}_{m})=\mathbf{h}_{m}$, and for $i=1,2,\cdots$, $\tilde{f}^{(i+1)}(\cdot)$ and $\tilde{h}^{(i+1)}(\cdot)$ are defined as follows:
\begin{align}
\begin{aligned}
&\tilde{f}^{(i+1)}(\mathbf{z}_{m}^{(1)},\cdots,\mathbf{z}_{m}^{(i+1)},\mathbf{h}_{m})\\
&=f^{(i)}(\mathbf{z}_{m}^{(2)}+\bar{\mu}_{m}^{(1)}\mathbf{h}_{m},\cdots,\mathbf{z}_{m}^{(i+1)}+\bar{\mu}_{m}^{(i)}\mathbf{h}_{m}),
\end{aligned}
\end{align}
\begin{align}
\tilde{g}^{(i+1)}(\mathbf{o}_{m}^{(1)},\cdots,\mathbf{o}_{m}^{(i)},\mathbf{n}_{m})=g^{(i)}(\mathbf{o}_{m}^{(2)},\cdots,\mathbf{o}_{m}^{(i)},\mathbf{o}_{m}^{(1)}+\mathbf{n}_{m}).
\end{align}

Then, the coefficients $\{a_{m,j}^{(i)}\}_{j=1}^{i}$ and $\{b_{m,j}^{(i)}\}_{j=1}^{i-1}$ are respectively given by the last rows of matrices $\hat{\mathbf{M}}^{(i)}_{m,a}$ and $\hat{\mathbf{M}}^{(i)}_{m,b}$ as follows:
\begin{align}
\hat{\mathbf{M}}^{(i)}_{m,a}=\sum_{j=0}^{i} \kappa_{2(j+1)} \hat{\mathbf{\Psi}}_{m}^{(i)}\left(\hat{\mathbf{\Phi}}_{m}^{(i)} \hat{\mathbf{\Psi}}_{m}^{(i)}\right)^j,
\label{Ma}
\end{align}
\begin{align}
\hat{\mathbf{M}}^{(i)}_{m,b}=\frac{L}{(T+1)N} \sum_{j=0}^{i-1} \kappa_{2(j+1)} \hat{\mathbf{\Phi}}_{m}^{(i)}\left(\hat{\mathbf{\Psi}}_{m}^{(i)} \hat{\mathbf{\Phi}}_{m}^{(i)}\right)^j.
\label{Mb}
\end{align}
\noindent In the above expressions, the auxiliary matrices $\hat{\mathbf{\Psi}}_{m}^{(i)}$ and $\hat{\mathbf{\Phi}}_{m}^{(i)}$ are defined as follows:
\begin{align}
\hat{\mathbf{\Psi}}_{m}^{(i)}=\left[\begin{array}{ccccc}
0 & 0 & \ldots & 0 & 0 \\
0 & \left\langle\partial_2 \mathbf{v}_{m}^{(2)}\right\rangle & 0 & \ldots & 0 \\
0 & \left\langle\partial_2 \mathbf{v}_{m}^{(3)}\right\rangle & \left\langle\partial_3 \mathbf{v}_{m}^{(3)}\right\rangle & \ldots & 0 \\
\vdots & \vdots & \vdots & \ddots & \vdots \\
0 & \left\langle\partial_2 \mathbf{v}_{m}^{(i)}\right\rangle & \left\langle\partial_3 \mathbf{v}_{m}^{(i)}\right\rangle & \ldots & \left\langle\partial_i \mathbf{v}_{m}^{(i)}\right\rangle
\end{array}\right],
\end{align}
\begin{align}
\hat{\mathbf{\Phi}}_{m}^{(i)}=\left[\begin{array}{ccccc}
0 & 0 & \ldots & 0 & 0 \\
\left\langle\partial_1 \mathbf{u}_{m}^{(2)}\right\rangle & 0 & 0 & \ldots & 0 \\
\left\langle\partial_1 \mathbf{u}_{m}^{(3)}\right\rangle & \left\langle\partial_2 \mathbf{u}_{m}^{(2)}\right\rangle & 0 & \ldots & 0 \\
\vdots & \vdots & \ddots & \vdots & \vdots \\
\left\langle\partial_1 \mathbf{u}_{m}^{(i)}\right\rangle & \left\langle\partial_2 \mathbf{u}_{m}^{(i)}\right\rangle & \ldots & \left\langle\partial_{i-1} \mathbf{u}_{m}^{(i)}\right\rangle & 0
\end{array}\right],
\end{align}
\noindent where $\partial_k \mathbf{v}_{m}^{(i)}$ and $\partial_k \mathbf{u}_{m}^{(i)}$ denote the row-wise partial derivatives.

On the other hand, the state evolution of the general recursion is defined as follows:
\begin{align}
\left[Z_{m}^{(1)}, \ldots, Z_{m}^{(i)}\right]^{\mathrm{T}} \sim \mathcal{CN}(\mathbf{0},\tilde{\mathbf{\Omega}}_{m}^{(i)}),
\label{generalSE1}
\end{align}
\begin{align}
V_{m}^{(i)}=\tilde{f}^{(i)}(Z_{m}^{(1)}, \ldots, Z_{m}^{(i)},H_{m}),
\label{generalSE2}
\end{align}
\begin{align}  
\left[O_{m}^{(1)},\cdots,O_{m}^{(i)}\right]^{\mathrm{T}} \sim \mathcal{CN}(\mathbf{0},\tilde{\mathbf{\Sigma}}_{m}^{(i)}),
\label{generalSE3}
\end{align}
\begin{align}
    U_{m}^{(i+1)}=\tilde{g}^{(i+1)}(O_{m}^{(1)},\cdots,O_{m}^{(i)},N_{m}).
\label{generalSE4}
\end{align}
\noindent The covariance matrices $\tilde{\mathbf{\Sigma}}_{m}^{(i+1)}$ and $\tilde{\mathbf{\Omega}}_{m}^{(i+1)}$ in (\ref{generalSE3}) and (\ref{generalSE1}) are computed as
\begin{align}
\tilde{\mathbf{\Sigma}}_{m}^{(i+1)} = \sum_{j=0}^{2i+1}\bar{\kappa}_{2(j+1)}\tilde{\mathbf{\Xi}}_{m,j}^{(i+1)},
\label{tildesigma}
\end{align}
\begin{align}
    \tilde{\mathbf{\Omega}}_{m}^{(i+1)} = \frac{L}{(T+1)N}\sum_{j=0}^{2i}\bar{\kappa}_{2(j+1)}\tilde{\mathbf{\Theta}}_{m,j}^{(i+1)},
\label{tildeomega}
\end{align}
where $\tilde{\mathbf{\Xi}}_{m,0}^{(i+1)}=\tilde{\mathbf{\Gamma}}_{m}^{(i+1)}$, $\tilde{\mathbf{\Theta}}_{m,0}^{(i+1)}=\tilde{\mathbf{\Delta}}_{m}^{(i+1)}$, and for $j \geq 1$:
\begin{align}
\begin{aligned}
\tilde{\mathbf{\Xi}}_{m,j}^{(i+1)} & =\sum_{k=0}^j\left(\tilde{\mathbf{\Psi}}_{m}^{(i+1)} \tilde{\mathbf{\Phi}}_{m}^{(i+1)}\right)^{k} \tilde{\mathbf{\Gamma}}_{m}^{(i+1)}\left(\left(\tilde{\mathbf{\Psi}}_{m}^{(i+1)} \tilde{\mathbf{\Phi}}_{m}^{(i+1)}\right)^{\mathrm{H}}\right)^{j-k} \\
& +\sum_{k=0}^{j-1}\left(\tilde{\mathbf{\Psi}}_{m}^{(i+1)} \tilde{\mathbf{\Phi}}_{m}^{(i+1)}\right)^k \tilde{\mathbf{\Psi}}_{m}^{(i+1)} \tilde{\mathbf{\Delta}}_{m}^{(i+1)} (\tilde{\mathbf{\Psi}}_{m}^{(i+1)})^{\mathrm{H}}\\
&\quad \quad \cdot\left(\left(\tilde{\mathbf{\Psi}}_{m}^{(i+1)} \tilde{\mathbf{\Phi}}_{m}^{(i+1)}\right)^{\mathrm{H}}\right)^{j-k-1},
\end{aligned}
\label{tildeXi}
\end{align}
\begin{align}
\begin{aligned}
\tilde{\mathbf{\Theta}}_{m,j}^{(i+1)} & =\sum_{k=0}^j\left(\tilde{\mathbf{\Phi}}_{m}^{(i+1)}\tilde{\mathbf{\Psi}}_{m}^{(i+1)} \right)^{k} \tilde{\mathbf{\Delta}}_{m}^{(i+1)}\left(\left(\tilde{\mathbf{\Phi}}_{m}^{(i+1)}\tilde{\mathbf{\Psi}}_{m}^{(i+1)} \right)^{\mathrm{H}}\right)^{j-k} \\
& +\sum_{k=0}^{j-1}\left(\tilde{\mathbf{\Phi}}_{m}^{(i+1)}\tilde{\mathbf{\Psi}}_{m}^{(i+1)} \right)^k \tilde{\mathbf{\Phi}}_{m}^{(i+1)} \tilde{\mathbf{\Gamma}}_{m}^{(i+1)} (\tilde{\mathbf{\Phi}}_{m}^{(i+1)})^{\mathrm{H}}\\
&\quad \quad \left(\left(\tilde{\mathbf{\Phi}}_{m}^{(i+1)}\tilde{\mathbf{\Psi}}_{m}^{(i+1)} \right)^{\mathrm{H}}\right)^{j-k-1}.
\label{tildetheta}
\end{aligned}
\end{align}

\noindent In particular, $\tilde{\mathbf{\Psi}}_{m,j}^{(i+1)}$ and $\tilde{\mathbf{\Phi}}_{m}^{(i+1)}$ are the deterministic versions of $\hat{\mathbf{\Psi}}_{m,j}^{(i+1)}$ and $\hat{\mathbf{\Phi}}_{m}^{(i+1)}$ given as follows:
\begin{align}
\begin{aligned}
\label{tildepsi}
&\tilde{\mathbf{\Psi}}_{m}^{(i+1)}=\\
&\left[\begin{array}{ccccc}
0 & 0 & \ldots & 0 & 0 \\
0 & \mathbb{E}[\partial_2 V_{m}^{(2)}] & 0 & \ldots & 0 \\
0 & \mathbb{E}[\partial_2 V_{m}^{(3)}] & \mathbb{E}[\partial_3 V_{m}^{(3)}] & \ldots & 0 \\
\vdots & \vdots & \vdots & \ddots & \vdots \\
0 & \mathbb{E}[\partial_3 V_{m}^{(i+1)}] & \mathbb{E}[\partial_3 V_{m}^{(i+1)}] & \ldots & \mathbb{E}[\partial_{i+1} V_{m}^{(i+1)}]
\end{array}\right],
\end{aligned}
\end{align}
\begin{align}
\begin{aligned}
\label{tildephi}
&\tilde{\mathbf{\Phi}}_{m}^{(i+1)}=\\
&\left[\begin{array}{ccccc}
0 & 0 & \ldots & 0 & 0 \\
\mathbb{E}[\partial_1 U_{m}^{(2)}] & 0 & 0 & \ldots & 0 \\
\mathbb{E}[\partial_1 U_{m}^{(3)}] & \mathbb{E}[\partial_2 U_{m}^{(2)}] & 0 & \ldots & 0 \\
\vdots & \vdots & \ddots & \vdots & \vdots \\
\mathbb{E}[\partial_1 U_{m}^{(i+1)}] & \mathbb{E}[\partial_2 U_{m}^{(i+1)}] & \ldots & \mathbb{E}[\partial_{i} U_{m}^{(i+1)}] & 0
\end{array}\right].
\end{aligned}
\end{align}
\noindent Besides, $\tilde{\mathbf{\Gamma}}_{m}^{(i+1)}$ and $\tilde{\mathbf{\Delta}}_{m}^{(i+1)}$ are given as follows:
\begin{align}
\begin{aligned}
&\tilde{\mathbf{\Gamma}}_{m}^{(i+1)}=\\
&\left[\begin{array}{cccc}
\mathbb{E}[(V_{m}^{(1)})^2] & \mathbb{E}[V_{m}^{(1)}V_{m}^{(2)}] & \ldots & \mathbb{E}[V_{m}^{(1)}V_{m}^{(i+1)}] \\
\mathbb{E}[V_{m}^{(1)}V_{m}^{(2)}] & \mathbb{E}[(V_{m}^{(2)})^2] & \ldots & \mathbb{E}[V_{m}^{(2)}V_{m}^{(i+1)}] \\
\mathbb{E}[V_{m}^{(1)}V_{m}^{(3)}] & \mathbb{E}[V_{m}^{(2)}V_{m}^{(3)}] & \ldots & \mathbb{E}[V_{m}^{(3)}V_{m}^{(i+1)}] \\
\vdots & \vdots & \ddots & \vdots \\
\mathbb{E}[V_{m}^{(1)}V_{m}^{(i+1)}] & \mathbb{E}[V_{m}^{(2)}V_{m}^{(i+1)}]  & \ldots & \mathbb{E}[(V_{m}^{(i+1)})^2]
\end{array}\right],
\label{Gammatilde}
\end{aligned}
\end{align}
\begin{align}
\begin{aligned}
&\tilde{\mathbf{\Delta}}_{m}^{(i+1)}=\\
&\left[\begin{array}{ccccc}
\mathbb{E}[(U_{m}^{(1)})^2] & \mathbb{E}[U_{m}^{(1)}V_{m}^{(2)}] & \ldots & \mathbb{E}[U_{m}^{(1)}U_{m}^{(i+1)}] \\
\mathbb{E}[U_{m}^{(1)}U_{m}^{(2)}] & \mathbb{E}[(U_{m}^{(2)})^2] & \ldots & \mathbb{E}[U_{m}^{(2)}U_{m}^{(i+1)}] \\
\mathbb{E}[U_{m}^{(1)}U_{m}^{(3)}] & \mathbb{E}[U_{m}^{(2)}U_{m}^{(3)}] & \ldots & \mathbb{E}[U_{m}^{(3)}U_{m}^{(i+1)}] \\
\vdots & \vdots & \ddots & \vdots \\
\mathbb{E}[U_{m}^{(1)}U_{m}^{(i+1)}] & \mathbb{E}[U_{m}^{(2)}U_{m}^{(i+1)}]  & \ldots & \mathbb{E}[(U_{m}^{(i+1)})^2]
\end{array}\right].
\label{Deltatilde}
\end{aligned}
\end{align}
The following theorem reveals the relationship between the general recursion and its state evolution.
\begin{theorem}
\label{theoremgeneral}
    Consider the general recursion in (\ref{generalrecursion1}) $\sim$ (\ref{generalrecursion4}) and its state evolution defined in (\ref{generalSE1}) $\sim$ (\ref{generalSE4}). Let $\tilde{\psi}$: $\mathbb{C}^{2i+1}\rightarrow \mathbb{C}$ and $\tilde{\phi}$: $\mathbb{C}^{2i+2}\rightarrow \mathbb{C}$ be any pseudo-Lipschitz functions of order 2. For $i = 1, 2, \cdots$, it is almost sure that
    \begin{align}
    \begin{aligned}
    \label{generaltheorem1}
        \lim _{(T+1)N \rightarrow \infty} &\frac{\sum_{t=1}^{(T+1)N}\tilde{\psi}\left(z_{t,m}^{(2)}, \ldots, z_{t,m}^{(i+1)}, v_{t,m}^{(2)}, \ldots, v_{t,m}^{(i+1)}, h_{t,m}\right)}{(T+1)N} \\
        &=\mathbb{E}\left[\tilde{\psi}\left(Z_{m}^{(2)}, \ldots, Z_{m}^{(i+1)}, V_{m}^{(2)}, \ldots, V_{m}^{(i+1)}, H_{m}\right)\right],
    \end{aligned}
    \end{align}
    \begin{align}
    \begin{aligned}
    \label{generaltheorem2}
        \lim _{L \rightarrow \infty} &\frac{1}{L} \sum_{l=1}^{L} \tilde{\phi}\left(o_{l,m}^{(1)}, \ldots, o_{l,m}^{(i)}, u_{l,m}^{(1)}, \ldots, u_{l,m}^{(i+1)}, n_{l,m}\right)\\
        &=\mathbb{E}\left[\tilde{\phi}\left(O_{m}^{(1)}, \ldots, O_{m}^{(i)}, U_{m}^{(1)}, \ldots, U_{m}^{(i+1)}, N_{m}\right)\right].
    \end{aligned}
    \end{align}
\end{theorem}
\noindent In other words, as $N,L \rightarrow \infty$, the following results hold almost surely:
\begin{align}
\begin{aligned}
&\left(\mathbf{z}_{m}^{(2)},\cdots,\mathbf{z}_{m}^{(i+1)},\mathbf{v}_{m}^{(2)},\cdots,\mathbf{z}_{m}^{(i+1)},\mathbf{h}_{m}\right)\\
&\stackrel{W_2}{\longrightarrow} \left(Z_{m}^{(2)}, \ldots, Z_{m}^{(i+1)}, V_{m}^{(2)}, \ldots, V_{m}^{(i+1)}, H_{m}\right),
\end{aligned}
\end{align}
\begin{align}
\begin{aligned}
&\left(\mathbf{o}_{m}^{(1)},\cdots,\mathbf{o}_{m}^{(i)},\mathbf{u}_{m}^{(1)},\cdots,\mathbf{u}_{m}^{(i+1)},\mathbf{n}_{m}\right)\\
&\stackrel{W_2}{\longrightarrow} \left(O_{m}^{(1)}, \ldots, O_{m}^{(i)}, U_{m}^{(1)}, \ldots, U_{m}^{(i+1)}, N_{m}\right).
\end{aligned}
\end{align}
\begin{proof}
    The proof is referred to Theorem 5.3 in \cite{zfan2020}.
\end{proof}

Theorem \ref{theoremgeneral} shows that the joint empirical distribution of $(\mathbf{z}_{m}^{(1)},\ldots,\mathbf{z}_{m}^{(i)})$ converges to an $i$-dimensional complex Gaussian distribution $\mathcal{CN}(\mathbf{0},\tilde{\mathbf{\Omega}}_{m}^{(i)})$, and the joint empirical distribution of $(\mathbf{o}_{m}^{(1)},\ldots,\mathbf{o}_{m}^{(i)})$ converges to an $i$-dimensional complex Gaussian distribution $\mathcal{CN}(\mathbf{0},\tilde{\mathbf{\Sigma}}_{m}^{(i)})$. Based on this theorem, we next claim that the state evolution of the general recursion and the state evolution of the FPAMP-based algorithm are equivalent.
\begin{lemma}
\label{lemmaSE=}
    For $i = 1, 2, \cdots$, we have $\tilde{\mathbf{\Sigma}}_{m}^{(i)}=\bar{\mathbf{\Sigma}}_{m}^{(i)}$ and $\tilde{\mathbf{\Omega}}_{m}^{(i+1)}=\check{\mathbf{\Omega}}_{m}^{(i+1)}$, where $\bar{\mathbf{\Omega}}_{m}^{(i)}$ is the lower right $i \times i$ submatrix of $\check{\mathbf{\Omega}}_{m}^{(i+1)}$.
\end{lemma}
\begin{proof}
We use induction to prove this lemma. In particular, for $i=1$, $O_{m}^{(1)} \sim \mathcal{CN}(0,\tilde{\mathbf{\Sigma}}_{m}^{(1)})$ with $\tilde{\mathbf{\Sigma}}_{m}^{(1)}=\bar{\kappa}_2\mathbb{E}[H_{m}^{2}]=\bar{\mathbf{\Sigma}}_{m}^{(1)}$. Then, matrix $\tilde{\mathbf{\Omega}}_{m}^{(2)}$ can be computed as
\begin{align}
\tilde{\mathbf{\Omega}}_{m}^{(2)}=\left[\begin{array}{cc}
0 & 0 \\
0 & \frac{L\left(\bar{\kappa}_2\mathbb{E}[(g^{(1)})^2]+\bar{\kappa}_4\mathbb{E}[(V_{m}^{(1)})^{2}](\mathbb{E}[\partial_{O_{m}^{(1)}}g^{(1)}])^2 \right)}{(T+1)N}
\end{array}\right].
\end{align}
Since $V_m^{(1)}$ and $H_{m}$ ($O_{m}^{(1)}$ and $G_{m}$) follow the same distribution, $\tilde{\mathbf{\Omega}}_{m}^{(2)}=\check{\mathbf{\Omega}}_{m}^{(2)}$. 

Suppose $\tilde{\mathbf{\Sigma}}_{m}^{(i)}=\bar{\mathbf{\Sigma}}_{m}^{(i)}$ and $\tilde{\mathbf{\Omega}}_{m}^{(i+1)}=\check{\mathbf{\Omega}}_{m}^{(i+1)}$ for $i \geq 1$. Also, based on the state evolution of the general recursion, we have $(O_{m}^{(1)},\!\cdots\!,O_{m}^{(i)})\!\sim \mathcal{CN}(\mathbf{0},\tilde{\mathbf{\Sigma}}_{m}^{(i)})$ and $(Z_{m}^{(1)},\cdots,Z_{m}^{(i+1)}) \sim \mathcal{CN}(\mathbf{0},\tilde{\mathbf{\Omega}}_{m}^{(i+1)})$. By utilizing the induction hypothesis $\tilde{\mathbf{\Sigma}}_{m}^{(i)}=\bar{\mathbf{\Sigma}}_{m}^{(i)}$ and $\tilde{\mathbf{\Omega}}_{m}^{(i+1)}=\check{\mathbf{\Omega}}_{m}^{(i+1)}$ in the definitions of $\tilde{\mathbf{\Psi}}_{m}^{(i+1)}$, $\tilde{\mathbf{\Phi}}_{m}^{(i+1)}$, $\tilde{\mathbf{\Gamma}}_{m}^{(i+1)}$, and $\tilde{\mathbf{\Delta}}_{m}^{(i+1)}$ in (\ref{tildepsi}) $\sim$ (\ref{Deltatilde}), we obtain $\tilde{\mathbf{\Psi}}_{m}^{(i+1)}=\bar{\mathbf{\Psi}}_{m}^{(i+1)}$, $\tilde{\mathbf{\Phi}}_{m}^{(i+1)}=\bar{\mathbf{\Phi}}_{m}^{(i+1)}$, $\tilde{\mathbf{\Gamma}}_{m}^{(i+1)}=\bar{\mathbf{\Gamma}}_{m}^{(i+1)}$, and $\tilde{\mathbf{\Delta}}_{m}^{(i+1)}=\bar{\mathbf{\Delta}}_{m}^{(i+1)}$. As a result, by comparing (\ref{Sigmaupdate}) and (\ref{tildesigma}), we have $\tilde{\mathbf{\Sigma}}_{m}^{(i+1)}=\bar{\mathbf{\Sigma}}_{m}^{(i+1)}$. This implies that $\tilde{\mathbf{\Phi}}_{m}^{(i+2)}=\bar{\mathbf{\Phi}}_{m}^{(i+2)}$, $\tilde{\mathbf{\Delta}}_{m}^{(i+2)}=\bar{\mathbf{\Delta}}_{m}^{(i+2)}$. Hence, according to (\ref{Omegaupdate2}) and (\ref{tildetheta}), we have $\tilde{\mathbf{\Theta}}_{m,j}^{(i+2)}=\mathbf{\Theta}_{m,j}^{(i+2)}$, and thus $\tilde{\mathbf{\Omega}}_{m}^{(i+2)}=\check{\mathbf{\Omega}}_{m}^{(i+2)}$, which completes the proof.
\end{proof}

Therefore, $(Z_m^{(2)},\cdots,Z_m^{(i+1)})$ and $(H_{m}^{(1)}-\bar{\mu}_{m}^{(1)}H_{m},\cdots,$ $H_{m}^{(i)}-\bar{\mu}_{m}^{(i)}H_{m})$ have the same distribution, which also applies to $(O_m^{(1)},\cdots,O_m^{(i)})$ and $(G_{m},R_{m}^{(1)},\cdots,R_{m}^{(i-1)})$. In order to prove Theorem \ref{theoremconverge}, it remains to show in the following lemma that the FPAMP iterations are close to the general recursion.
\begin{lemma}
\label{closelemma}
    Let $\psi$: $\mathbb{C}^{2i+1}\rightarrow \mathbb{C}$ and $\phi$: $\mathbb{C}^{2i+2}\rightarrow \mathbb{C}$ be any pseudo-Lipschitz functions of order 2. For $i = 1, 2, \cdots$, it is almost sure that
    \begin{align}
    \begin{aligned}
    \label{close1}
        &\lim _{(T+1)N \rightarrow \infty} \frac{1}{(T+1)N}\\ &\Bigg|\sum_{t=1}^{(T+1)N} \psi\left(h_{t,m}^{(1)}, \ldots, h_{t,m}^{(i)}, \tilde{h}_{t,m}^{(1)}, \ldots, \tilde{h}_{t,m}^{(i)}, h_{t,m}\right)\\
        &-\sum_{t=1}^{(T+1)N}\psi\Big(z_{t,m}^{(2)}+\bar{\mu}_{m}^{(1)}h_{t,m}, \ldots, z_{t,m}^{(i+1)}+\bar{\mu}_{m}^{(i)}h_{t,m},\\
        &\quad \quad \quad \quad \quad v_{t,m}^{(2)}, \ldots, v_{t,m}^{(i+1)}, h_{t,m}\Big)\Bigg|=0,
    \end{aligned}
    \end{align}
    \begin{align}
    \begin{aligned}
    \label{close2}
        &\lim_{L \rightarrow \infty} \frac{1}{L} \Bigg|\sum_{l=1}^{L} \phi\left(r_{l,m}^{(1)}, \ldots, r_{l,m}^{(i)}, s_{l,m}^{(1)}, \ldots, s_{l,m}^{(i+1)}, y_{l,m}\right)\\
        &-\sum_{l=1}^{L}\phi \left(o_{l,m}^{(2)}, \ldots, o_{l,m}^{(i+1)}, u_{l,m}^{(2)}, \ldots, u_{l,m}^{(i+2)}, o_{l,m}^{(1)}+n_{l,m}\right)\Bigg|=0.
    \end{aligned}
    \end{align}
\end{lemma}
\begin{proof}
By applying Cauchy-Schwarz inequality, we have the following results since $\psi\left(\cdot\right)$ is a pseudo-Lipschitz function: 
\begin{align}
\begin{aligned}
\label{ineq1}
    &\frac{1}{(T+1)N} \Bigg|\sum_{t=1}^{(T+1)N} \psi\left(h_{t,m}^{(1)}, \ldots, h_{t,m}^{(i)}, \tilde{h}_{t,m}^{(1)}, \ldots, \tilde{h}_{t,m}^{(i)}, h_{t,m}\right)\\
        &-\sum_{t=1}^{(T+1)N}\psi\Big(z_{t,m}^{(2)}+\bar{\mu}_{m}^{(1)}h_{t,m}, \ldots, z_{t,m}^{(i+1)}+\bar{\mu}_{m}^{(i)}h_{t,m},\\
        &\quad \quad \quad \quad \quad v_{t,m}^{(2)}, \ldots, v_{t,m}^{(i+1)}, h_{t,m}\Big)\Bigg|\\
        &\leq \frac{C_0}{(T+1)N}\sum_{t=1}^{(T+1)N}\Bigg(1+|h_{t,m}|+\sum_{j=1}^{i}\Big(|h_{t,m}^{(j)}|+|\tilde{h}_{t,m}^{(j)}|+|z_{t,m}^{(j+1)}\\
        &\quad \quad \quad \quad \quad \quad \quad \quad \quad+\bar{\mu}_{m}^{(j)}h_{t,m}|+|v_{t,m}^{(j+1)}|\Big)\Bigg)\\
        &\quad \times\left(\sum_{j=1}^{i}\left(|h_{t,m}^{(j)}-z_{t,m}^{(j+1)}-\bar{\mu}_{m}^{(j)}h_{t,m}|^2+|\tilde{h}_{t,m}^{(j)}-v_{t,m}^{(j+1)}|^2\right)\right)^{\frac{1}{2}}\\
        &\leq C_0(4i+2)\Bigg[1+\frac{||\mathbf{h}_{m}||^2}{(T+1)N}\\
        &\ +\sum_{j=1}^{i}\frac{||\mathbf{h}_{m}^{(j)}||^2+||\tilde{\mathbf{h}}_{m}^{(j)}||^2+||\mathbf{z}_{m}^{(j+1)}+\bar{\mu}_{m}^{(j)}\mathbf{h}_{m}||^2+||\mathbf{v}_{m}^{(j+1)}||^2}{(T+1)N}\Bigg]^{\frac{1}{2}}\\
        &\quad \times \left(\sum_{j=1}^{i}\frac{||\mathbf{h}_{m}^{(j)}-\mathbf{z}_{m}^{(j+1)}-\bar{\mu}_{m}^{(j)}\mathbf{h}_{m}||^2+||\tilde{\mathbf{h}}_{m}^{(j)}-\mathbf{v}_{m}^{(j+1)}||^2}{(T+1)N}\right)^{\frac{1}{2}},
\end{aligned}
\end{align}
where $C_0$ denotes a generic positive constant. Similarly, we have
\begin{align}
\begin{aligned}
\label{ineq2}
&\frac{1}{L} \Bigg|\sum_{l=1}^{L}\phi\left(r_{l,m}^{(1)}, \ldots, r_{l,m}^{(i)}, s_{l,m}^{(1)}, \ldots, s_{l,m}^{(i+1)}, y_{l,m}\right)\\
&-\sum_{l=1}^{L}\phi \left(o_{l,m}^{(2)}, \ldots, o_{l,m}^{(i+1)}, u_{l,m}^{(2)}, \ldots, u_{l,m}^{(i+2)}, o_{l,m}^{(1)}+n_{l,m}\right)\Bigg|\\
&\leq C_{0}(4i+5)\Bigg[1+\frac{||\mathbf{y}_{m}||^2+||\mathbf{o}_{m}^{(1)}+\mathbf{n}_{m}||^2}{L}\\
&+\sum_{j=1}^{i}\frac{||\mathbf{r}_{m}^{(j)}||^2+||\mathbf{o}_{m}^{(j+1)}||^2}{L}+\sum_{j=1}^{i+1}\frac{||\mathbf{s}_{m}^{(j)}||^2+||\mathbf{u}_{m}^{(j+1)}||^2}{L}\Bigg]^{\frac{1}{2}}\\
&\times \Bigg(\frac{||\mathbf{s}_{m}^{(i+1)}-\mathbf{u}_{m}^{(i+2)}||^2+||\mathbf{y}_{m}-\mathbf{o}_{m}^{(1)}-\mathbf{n}_{m}||^2}{L}\\
&+\sum_{j=1}^{i}\frac{||\mathbf{r}_{m}^{(j)}-\mathbf{o}_{m}^{(j+1)}||^2+||\mathbf{s}_{m}^{(j)}-\mathbf{u}_{m}^{(j+1)}||^2}{L}\Bigg)^{\frac{1}{2}}.
\end{aligned}
\end{align}
To prove this lemma, we just need to show that as $N,L\rightarrow \infty$, the terms inside the brackets of (\ref{ineq1}) and (\ref{ineq2}) converge to finite and deterministic limits while the terms in the last line of (\ref{ineq1}) and (\ref{ineq2}) converge to zero. This can be accomplished via mathematical induction following Appendix D.4 in \cite{marco2022}.

\end{proof}

Let
\begin{align}
\begin{aligned}
    &\tilde{\psi}\Big(z_{t,m}^{(2)}, \ldots, z_{t,m}^{(i+1)}, v_{t,m}^{(2)}, \ldots, v_{t,m}^{(i+1)}, h_{t,m}\Big)\\
    &=\psi\Big(z_{t,m}^{(2)}+\bar{\mu}_{m}^{(1)}h_{t,m}, \ldots, z_{t,m}^{(i+1)}+\bar{\mu}_{m}^{(i)}h_{t,m},\\
    &\quad \quad v_{t,m}^{(2)}, \ldots, v_{t,m}^{(i+1)}, h_{t,m}\Big).
\end{aligned}
\end{align}
It is almost sure that
\begin{align}
    \begin{aligned}
    \label{transform}
        &\lim _{(T+1)N \rightarrow \infty} \frac{1}{(T+1)N}\\ &\sum_{t=1}^{(T+1)N} \psi\Big(z_{t,m}^{(2)}+\bar{\mu}_{m}^{(1)}h_{t,m}, \ldots, z_{t,m}^{(i+1)}+\bar{\mu}_{m}^{(i)}h_{t,m}, \\ 
        &\quad \quad \quad \quad v_{t,m}^{(2)}, \ldots, v_{t,m}^{(i+1)}, h_{t,m}\Big)\\
        &\overset{(\text{a})}{=}\mathbb{E}\Big[\psi\Big(Z_{m}^{(2)}+\bar{\mu}_{m}^{(1)}H_{m}, \ldots, Z_{m}^{(i+1)}+\bar{\mu}_{m}^{(i)}H_{m},\\
        &\quad \quad \quad \quad V_{m}^{(2)}, \ldots, V_{m}^{(i+1)}, H_{m}\Big)\Big]\\
        &\overset{(\text{b})}{=}\mathbb{E}\left[\psi\left(H_{m}^{(1)}, \ldots, H_{m}^{(i)}, \tilde{H}_{m}^{(1)}, \ldots, \tilde{H}_{m}^{(i)}, H_{m}\right)\right],
    \end{aligned}
\end{align}
\noindent where (a) can be obtained via (\ref{generaltheorem1}) in Theorem \ref{theoremgeneral}, and (b) is based on Lemma \ref{lemmaSE=}, (\ref{SE4}), and (\ref{generalSE2}). Therefore, with (\ref{transform}) and (\ref{close1}) of Lemma \ref{closelemma}, we can obtain (\ref{theoremconverge1}) of Theorem \ref{theoremconverge}. Meanwhile, (\ref{theoremconverge2}) can be derived similarly, which completes the proof of Theorem \ref{theoremconverge}.

\ifCLASSOPTIONcaptionsoff
  \newpage
\fi

\end{document}